\documentclass[preprint,12pt]{elsarticle}

\usepackage{graphicx}				
\usepackage{amsmath,amssymb}
\usepackage{amsthm}

\usepackage{caption}
\usepackage{subcaption}
\usepackage{csquotes}
\usepackage{tikz}
\usetikzlibrary{arrows,
                backgrounds,
                fit,
                positioning,
                quotes,
                shapes}
\usepackage{comment}

\usetikzlibrary{shapes.geometric}
\usetikzlibrary{arrows}
\usetikzlibrary{arrows.meta}
\usetikzlibrary{calc}
\usetikzlibrary{positioning}
\usetikzlibrary{decorations.text}
\usetikzlibrary{decorations.pathmorphing}
\usetikzlibrary{hobby}

\usepackage{natbib}
\usepackage[section]{placeins}
\usepackage[parfill]{parskip}
\usepackage{paralist}
\usepackage{graphicx}
\usepackage{wrapfig}
\usepackage{lscape}
\usepackage{rotating}
\usepackage{epstopdf}
\usepackage{algorithm}
\usepackage[noend]{algpseudocode}

\def\equationautorefname~#1\null{%
  (#1)\null
}

\makeatletter
\newenvironment{breakablealgorithm}
{
    \begin{center}
        \refstepcounter{algorithm}
        \hrule height.8pt depth0pt \kern2pt
        \renewcommand{\caption}[2][\relax]{
            {\raggedright\textbf{\ALG@name~\thealgorithm} ##2\par}%
            \ifx\relax##1\relax 
                \addcontentsline{loa}{algorithm}{\protect\numberline{\thealgorithm}##2}%
            \else 
                \addcontentsline{loa}{algorithm}{\protect\numberline{\thealgorithm}##1}%
            \fi
            \kern2pt\hrule\kern2pt
}
        }{
        \kern2pt\hrule\relax
    \end{center}
}
\makeatother

\usepackage[]{color}

\usepackage[colorlinks=true,linkcolor=blue,citecolor=blue,urlcolor=blue,
            pdfauthor={},pdftitle={}]{hyperref}

\newtheorem{definition}{Definition}[section]
\newtheorem{theorem}{Theorem}[section]
\newtheorem{lemma}{Lemma}[section]

\newcommand{\eg}{\emph{e.g., }}
\newcommand{\ie}{\emph{i.e., }}

\title{Transitivity Preserving Projection in Directed Hypergraphs}
\author{Eric Parsonage}
\author{Matthew Roughan}
\author{Hung X Nguyen}
\affiliation{organization={The University of Adelaide},
            addressline={School of Mathematical and Computer Sciences},
            city={Adelaide, South Australia},
            country={Australia}}
\date{}							

\begin{document}

\begin{abstract}
Directed hypergraphs are vital for modeling complex polyadic relationships in domains such as discrete mathematics, computer science, network security, and systems modeling. However, their inherent complexity often impedes effective visualization and analysis, particularly for large graphs. This paper introduces a novel Transitivity Preserving Projection (TPP) to address the limitations of the computationally intensive Basu and Blanning projection (BBP), which can paradoxically increase complexity by flattening transitive relationships. TPP offers a minimal and complete representation of relationships within a chosen subset of elements, capturing only irreducible dominant metapaths to ensure the smallest set of edges while preserving all essential transitive and direct connections. This approach significantly enhances visualization by reducing edge proliferation and maintains the integrity of the original hypergraph's structure. We develop an efficient algorithm leveraging the set-trie data structure, reducing the computational complexity from  an exponential number of metapath searches in BBP to a linear number of metapath searches with polynomial-time filtering, enabling scalability for real-world applications. Experimental results demonstrate TPP's superior performance, completing projections in seconds on graphs where BBP fails to terminate within 24 hours. By providing a minimal yet complete view of relationships, TPP supports applications in network security and supply chain analysis, offering a clearer, more efficient framework for hypergraph simplification and analysis.


\end{abstract}

\maketitle
\pagenumbering{arabic} 

\section{Introduction}

A classical graph (or network) does not lend itself well to modeling polyadic relationships between sets of elements. As a result, more generalized graphical structures have been developed, including metagraphs and directed hypergraphs \citep{gallo1993directed,doi:10.1089/cmb.2023.0242,10.1007/3-540-45446-2_20,math10050759}. Hypergraphs are similar to metagraphs - they are both a generalization of classical graphs where sets of elements are connected by a single hyperedge or metaedge. Structurally they are the same, the key differences between metagraphs and hypergraphs are in the way paths are defined on these graphs, as explained in~\citet{parsonage2024}.  Directed hypergraphs have been extensively studied, but there are concepts developed under the heading of metagraph theory that have yet to be fully incorporated into the domain of hypergraphs. One such concept is that of projection.


In situations where a metagraph or directed hypergraph has more than a few edges and/or elements, the resulting complexity can impede the visualization advantages associated with such a structure~\citep{9721603,10290912,10673783}. To address this challenge, it is beneficial to concentrate on smaller subsets of elements and their corresponding relationships by condensing the structure into a more succinct metagraph representation~\citep{10290912,10673783}. In their work \textit{Metagraphs and their Applications} \citet{basu2007metagraphs} introduce the concept of projection as a means to condense a graph with the stated primary objective of abstracting away extraneous details and accentuating solely the relationships among the elements on which the projection is focused.  Moreover, Basu and Blanning's projection flattens out the relationships between the members of the generating sets over which we project so transitive relations are not explicitly obvious. This removes structure and increases the difficulty of visualization because the projection can have more edges than the metagraph it is ``simplifying."



We develop alternative approaches to projection that both demonstrate significantly improved performance even in Basu and Blanning's test cases and provide a condensed view of the network that (i) preserves important (and well-defined) structure, and (ii) has a much more feasible implementation. The result is the Transitivity Preserving Projection (TPP). Proofs are provided to show that the proposed TPP exhibits the necessary properties to fulfil its intended purpose, in particular that it produces a minimal representation for particular classes of directed hypergraphs that has applications beyond just simple visualization~\citep{ausiello1986}.
We also develop an algorithm to perform the projection and show the time complexity is significantly less than that of Basu and Blanning's projection.





Our main contributions in this paper are

\begin{enumerate}
    \item  \textbf{A novel Transitivity Preserving Projection (TPP)} method for directed hypergraphs (and metagraphs) that preserves transitive relationships while reducing structural complexity. We provide rigorous definitions and proofs establishing that TPP is:
    \begin{itemize}
        \item \textit{Unique} — each projection is well-defined and deterministic.
        \item \textit{Idempotent} — applying the projection multiple times yields the same result.
        \item \textit{Minimal} — it contains the smallest set of edges necessary to preserve all dominant relationships.
    \end{itemize}

    \item \textbf{Scalable projection algorithms} that leverage the efficiency of \textit{set-trie} data structures for rapid subset and superset queries. Our algorithms apply \textit{edge combining} techniques and \textit{breadth-first search} to minimize the search space, and systematically filter \textit{dominant and irreducible metapaths} to construct minimal projections. Unlike prior approaches, our TPP algorithm reduces the number of required metapath searches from exponential (in the size of the projection set) to linear, offering a significant improvement in computational tractability.

      \item \textbf{Improved visualization and structural clarity of directed hypergraphs}:
The Transitivity Preserving Projection (TPP) enhances the interpretability of complex metagraphs by preserving transitive relationships while minimizing edge redundancy. Unlike traditional projection methods that often obscure structural insights by flattening or over-representing pairwise connections, TPP maintains the logical flow of dependencies inherent in the original system. This results in a more concise and semantically meaningful representation, which facilitates clearer visual analysis. Such clarity is particularly valuable in domains where understanding system behavior is critical—such as network security, infrastructure modeling, and organizational analysis—enabling more effective communication, hypothesis generation, and decision-making.

\end{enumerate}

We evaluate TPP across a range of metagraph examples and demonstrate its substantial advantages in both information condensation and computational efficiency. For instance, TPP processes moderate-sized graphs in seconds, whereas the Basu and Blanning algorithm fails to complete within a reasonable timeframe. TPP also scales effectively to large, real-world networks—such as those found in supply chains, cybersecurity infrastructures, and university systems -- where traditional methods are hindered by memory or performance limitations. The algorithms are implemented in a reusable and extensible Python  framework, enabling practical adoption and further research.






\section{Background}


\subsection{Hypergraphs and metagraphs}

A metagraph (as defined in \citet{basu2007metagraphs}) is a generalised graph-theoretic structure that can be described as a collection of directed set-to-set mappings. 

Generally, graphs usually have two components, vertices and edges. Metagraphs, however, are described using three: a generating set $X$, a set of edges\footnote{For conciseness we refer to hyperedges and metaedges simply as edges, and likewise for vertices.} $E$, and vertices. Since vertices are nothing but subsets of $X$ that appear in some edge, this division is somewhat arbitrary. More formally: 

\begin{definition}[Generating set]
    \label{def:generatingset}
    \sloppy The variables of interest in a metagraph form the set of elements called the generating set $\mathit{X =\{x_1, x_2, x_3, \ldots\}}$.
\end{definition}

\begin{definition}[Edge]
    \label{def:edge}
    In a metagraph an edge $\mathit{e}$ on the generating set $X$ is a pair $e = \langle V_e, W_e \rangle$ consisting of an invertex $\mathit{V_e \subset X}$ and an outvertex $\mathit{W_e \subset X}$. We further specify that a $B$-edge is such that $|W_e|=1$, \ie a many-to-1 mapping, and an $F$-edge where $|V_e|=1$, \ie a 1-to-many mapping. 
\end{definition}

\begin{definition}[Metagraph]
    \label{def:metagraph}
    A metagraph $\mathit{\mathcal{H}= \langle X, E \rangle}$ is a structure specified by its generating set $\mathit{X}$ and a set of edges $\mathit{E}$ on $X$.
\end{definition}

See \autoref{fig:mg} for an example metagraph, which shows subsets of a generating set $X = \{ x_1, \ldots, x_9\}$ connected by a set of mappings (edges) $E=\{e_1, e_2, e_3, e_4\}$ noting that empty subsets are allowed.


\begin{figure}[htb]
\centering
\includegraphics[width=0.7\textwidth]{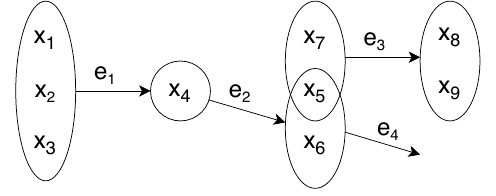}
\caption{A metagraph.
The generating set of this metagraph is $\{x_1, x_2, \ldots, x_9\}$. The edge $e_2$ can be described by the pair $\langle \{x_4\}, \{x_5, x_6\} \rangle$, which are respectively the invertex and outvertex of the edge.
An example metapath is $M(\{x_4, x_7\}, \{ x_8\}) = \langle e_2, e_3 \rangle$.
}
\label{fig:mg}
\end{figure}

Directed hypergraphs\footnote{In this paper, we are concerned with directed hypergraphs and thus refer to directed hypergraphs simply as hypergraphs.} are almost identical to metagraphs~\cite{parsonage2024}. Like metagraphs, each edge in a hypergraph is a directed set-to-set mapping, where the source set is called the tail of the edge and the target set is called the head of the edge. Thus a hypergraph can be defined as follows:

\begin{definition}[Hypergraph]
\label{def:Hypergraph}
A directed hypergraph is a pair $H=(V,E)$, where $V=\{v_1,v_2,\ldots,v_n\}$ is the set of vertices, and $E= \{e_1,e_2,\ldots,e_m\}$ is the set of edges. A hypergraph composed purely of B-edges (F-edges) is called a B-hypergraph (F-hypergraph).
\end{definition}

\begin{figure}[htb]
\centering
\includegraphics[width=0.7\textwidth]{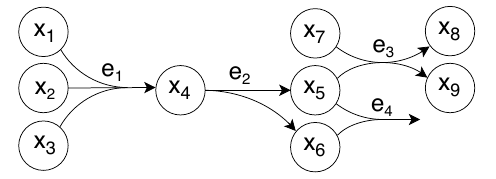}
\caption{A hypergraph representation of the metagraph shown in \autoref{fig:mg}.
Edges can have multiple sources, \eg edge $e_1$, multiple targets, \eg edge $e_2$ or both, \eg edge  $e_3$. Edges can also have either an empty head or tail, \eg edge  $e_4$.
}
\label{fig:hypergraph}
\end{figure}

\autoref{fig:hypergraph} represents a hypergraph. Edges can have multiple sources, \eg edge $e_1$ (which is a B-edge), multiple targets, \eg edge $e_2$ (which is an F-edge), or both, \eg edge $e_3$. Edges can also have either an empty head or tail, \eg edge $e_4$. Note that as a convention we visualise metagraphs and hypergraphs differently, so the hypergraph of \autoref{fig:hypergraph} is equivalent to the metagraph of \autoref{fig:mg}.


A key concept used here is that of transitive closure. In simple graph settings, transitive closure represents a completion of the graph such that transitive relationships are preserved. For instance, if there is an edge from A to B, and from B to C, a closure might highlight the fact that there will be a path from A to C by including an edge from A to C. In general such closures can be applied with respect to an arbitrary relation, inclusive of (abstract) weights on edges and so the concept is quite general. We will see that projection is intimately linked to the concept of transitive closure --- transitivity is a desirable property to preserve in any simplification. For instance, it is desirable that a projection (that condenses some part of a metagraph) preserve those transitive properties of the input that are still accessible in the output. 

However, transitive closure has not been uniquely defined for metagraphs. There is a direct connection to the succinct transitive closure of a B-Hypergraph, as described in \citep{ausiello2016}, however, that has the clear limitation of applying only to a subset of cases. So in later parts of the paper we seek to link up to, and extend ideas expressed in \citep{ausiello2016}.



\subsection{Paths in metagraphs}

A core concept in the theory of simple graphs is that of a {\em path.} This is even more true in the study of metagraphs, but there are now two non-equivalent ideas under the heading of paths. First a simple path generalises the idea of path in the most obvious manner by considering that an edge in a metagraph expresses a relationship betwen every member of the invertex to every member of the outvertex. Thus, simple paths are defined as follows:

\begin{definition}[Simple path]
    \label{def:simple-path}
    A simple-path is a sequence of edges $\langle e_1,\allowbreak\ e_2,\allowbreak\ \ldots,\allowbreak\ e_n \rangle$
 from element $\mathit{x}$ to element $\mathit{y}$ with $\mathit{x \in invertex(e_1),\allowbreak\ y \in outvertex(e_n)}$ and $\mathit{\forall \; e_i, i = 1,\ldots,n-1}$, we have that $\mathit{outvertex(e_i)}$ $\mathit{\cap}$ $\mathit{invertex \allowbreak (e_{i+1}) \neq \emptyset}$.
\end{definition}

It is arguable that the simple path definition provides no additional value for metagraphs over simple graphs as it effectively collapses the metagraph into its simple elements. A more important concept then is the idea of a metapath as defined below: 

\begin{definition}[Metapath]
    \label{def:metapath}
    Given a metagraph $\mathit{\mathcal{H}= \langle X, E \rangle}$, a metapath $\mathit{M(B, C)}$ from a source $\mathit{B \subset X}$ to a target $\mathit{C \subset X}$ is a set of edges $\mathit{E' \subseteq E}$ such that
    \begin{enumerate}
        \item every $\mathit{e \in E'}$ is on a simple-path from some element in $\mathit{B}$ to some element in $\mathit{C}$ and,
        \item $\mathit{\forall e = \langle V_e, W_e \rangle \in E', \bigcup_{e}V_{e} \setminus \bigcup_{e}W_{e} \subseteq B}$, (the source includes all the pure inputs on the metapath) and,
        \item $\mathit{C \subseteq \bigcup_{e}W_{e}}$, (the target is included in (a subset of) the union of outvertices).
    \end{enumerate}
\end{definition}

One way to intuitively think about a metapath lies in a possible interpretation of an edge as a logical relation. For instance, we might interpret an edge $( \{x_1,x_2\}, \{x_4\} )$ as saying that $x_4$ is enabled by $x_1 \mbox{ AND } x_2$. A metapath then defines a set of edges and nodes that extends this idea to provide 'path' enabling the outset of the metapath. 

Metapaths are quite different from simple paths: (i) metapaths are a {\em unordered} set of edges, whereas a simple path is a sequence, and (ii) a metapath describes a relationship between two sets.  

An example metapath in \autoref{fig:mg}, would be $M(\{x_1, x_2, x_3\}, \{ x_6\}) = \{ e_1, e_2 \}$, which says that edges $e_1$ and $e_2$ and inputs $\{x_1, x_2, x_3\}$ enable the output $x_6$. A counter-example is the lack of metapath from $\{x_1, x_2, x_3\}$ to $x_9$ because there is no means for the input elements to enable $x_7$, which is needed to enable $x_9$. 

As metapaths are sets not sequences, we cannot easily define a``shortest metapath," but there is an analogous concept of minimality called {\em dominance} as defined as follows:

\begin{definition}[Edge-dominant]
    Given a metagraph $\mathit{\mathcal{H} = \langle X, E \rangle}$, for any two sets of elements $\mathit{B,C \subset X}$, a metapath $\mathit{M(B, C)=E'}$ is said to be edge-dominant if there is no proper subset of edges from $\mathit{E'}$ forming a metapath from $\mathit{B}$ to $\mathit{C}$.
    \label{def:edge-dominant}
\end{definition}

\begin{definition}[Input-dominant]
    Given a metagraph $\mathit{\mathcal{H} = \langle X, E \rangle}$, for any two sets of elements $B,C \subset X$, a metapath $M(B, C)$ is said to be input-dominant if there does not exist any metapath $M'(B', C)$ with $B' \subset B$.
    \label{def:input-dominant}
\end{definition}

\begin{definition}[Dominant]
    \sloppy A metapath is dominant if it is both input-dominant and edge-dominant.
    \label{def:dominant}
\end{definition}



The definition of metapath above comes from \citet{basu2007metagraphs}, whose definition is directly related to their algorithm for finding dominant paths. We can, in fact, remove clause 1 from Definition~\ref{def:metapath} without changing results, such as the set of dominant paths. 


\subsection{Existing projection algorithm}

Basu and Blanning's Projection (BBP), \citep{basu2007metagraphs, blanning1997} is the state of the art for condensing the information of a metagraph. Their projection is defined as follows:

\begin{definition}[Basu and Blanning's Projection]
    \label{projection}
    Given a metagraph $\mathcal{H}=\langle\mathit{X,E}\rangle$ and $\mathit{X'}
        \subseteq\mathit{X}$, a metagraph $\mathcal{H'}=\langle\mathit{X',E'}\rangle$ is a projection of $\mathcal{H}$ over $\mathit{X'}$ if:
    \begin{enumerate}
        \item For any $\mathit{e'}=\langle\mathit{V',W'}\rangle\in\mathit{E'}$ and for any $\mathit{x'}\in\mathit{W'}$ there is a dominant metapath $\mathit{M(V',\{x'\})}$ in $\mathcal{H}$ and
        \item For every $\mathit{x'\in X'}$, if there is any dominant metapath $\mathit{M(V,\{x'\})}$ in $\mathcal{H}$ with $\mathit{V}\subseteq\mathit{X'}$, then there is an edge $\langle\mathit{V',W'}\rangle\in\mathit{E'}$ such that $\mathit{V'=V}$ and $\mathit{x'}\in\mathit{W'}$, and
        \item No two edges in $\mathit{E'}$ have the same invertex.
    \end{enumerate}
    \label{def:projection}
\end{definition}

The intention of the BBP is visible in that definition. Part 1 expresses that we do not introduce new relationships through the projection, and Part 2, expresses that we preserve metapaths present in the original metagraph as edges in the new graph. Thus if applied such that $X' = X$, it would represent a type of transitive closure. 

Intuitively speaking, a transitive closure of a graph is a `complexification' of the graph in the sense that if we were to take a connected graph, the transitive closure would increase the number of edges to $O(n^2)$ up from commonly sparse graphs where the number of edges might be $O(n)$ or $O(n \log n)$. 

The reason that BBP might be considered a simplification is that we project onto a reduced generating set $X'$ which, in most applications, would be a (small) proper subset of $X$. However we can do better. 

To add formally to this description we need to define projection reverse mapping, which is an operation we shall use later in any case. 
\begin{definition}[Projection Reverse Mapping]
    Given a metagraph $\mathcal{H}=\langle\mathit{X,E}\rangle$ and its projection $\mathcal{H'}=\langle\mathit{X',E'}\rangle$ over $\mathit{X'}\subseteq\mathit{X}$ then for any $\mathit{e'}=\langle\mathit{V',W'}\rangle\in\mathit{E'}$ the set $\mathit{C(e')}$ is the set of metapaths in $\mathcal{H}$ that correspond to $\mathit{e'}$.
    \label{def:composition}
\end{definition}

Now we can formally state that the  BPP is a surjection from dominant metapaths in $\mathcal{H}$ to edges in $\mathcal{H'}$ and the reverse map $\mathit{C(e')}$ of an edge $\mathit{e'}$ in $\mathcal{H'}$ is the preimage of $\mathit{e'}$ under projection.


\subsection{Existing Path Finding Algorithms}

The existing algorithms are due to \citet{blanning1997} as republished in the book \textit{Metagraphs and their Applications} \citep{basu2007metagraphs}. There are two major algorithms in the book, one to find metapaths (Definition~\ref{def:metapath}) between source and target sets and another to build a projection (Definition~\ref{def:projection}). Nearly all metagraph algorithms use the path finding algorithm thus its efficiency is critical.





Basu and Blanning's algorithm for finding paths between a source set and a target starts by building an adjacency matrix $A$ of edges and then calculating its transitive closure $A^*=A + A^2 + A^3 + \ldots + A^n$, such that $A^*_{i,j}$  has all the simple-paths between element $i$ and element $j$ in the $(i,j)$th entry. This is equivalent to a brute-force calculation of all shortest paths on a simple graph, and we will show that we can do better.

\subsection{Existing Projection Algorithms}

It is important that such a transformation correctly maintains relationships between elements of the transformation, even though the full set of edges between them may not be visible. Thus, it is crucial that such a transformation is rigorously defined and has provable properties. To this end \citet{blanning1997} and \citet{basu2007metagraphs} provide proofs of the uniqueness of the resulting projections.

In both those works the authors present an example metagraph (isomorphic to \autoref{fig:motivation}).
Here we conduct a forensic examination of this example and discover an authorial mistake in its projection calculation. By extending the example, we demonstrate how this error obscures the underlying complexity of the projection model. We also show a construction that will result in a countably infinite number of cases where the projection contains more edges and nodes than the initial metagraph, contradicting the intended purpose of providing a simplified view.

\begin{figure}[!p]
    \centering
    \begin{subfigure}{0.95\linewidth}
    \hspace*{0.02\linewidth}
     \begin{center}
     
\resizebox{0.7\textwidth}{!}{
    \begin{tikzpicture}
        \node[minimum width=1cm,draw,circle] (1) at (0,3.5) {$\mathit{x_1}$};
        \node[minimum width=1cm,draw,circle] (2) at (0,1){$\mathit{x_2}$};

        \node[minimum width=1cm,draw,circle] (3) at (2.5,4) {$\mathit{x_3}$};
        \node (4) at (2.5,2.5) {$\mathit{x_4}$};
        \node[minimum width=1cm,draw,circle] (5) at (2.5,1){$\mathit{x_5}$};

        \node[ellipse, minimum width = 1.8cm, minimum height = 3cm, draw] (3_4) at (2.5,3.5){};

        \node[ellipse, minimum width = 1.8cm, minimum height = 3cm, draw] (4_5) at (2.5,1.5){};

        \node[minimum width=1cm,draw,circle] (6) at (5,3.5) {$\mathit{x_6}$};
        \node[minimum width=1cm,draw,circle] (7) at (5,1.5){$\mathit{x_7}$};
        \node[ellipse, minimum width = 2cm, minimum height = 4cm, draw] (6_7) at (5,2.5){};
        \node[minimum width=1cm,draw,circle] (8) at (7.5,2.5){$\mathit{x_8}$};
        \draw[-{Latex[length=2mm]}] (1) -- (3_4) node [midway, above, sloped] (TextNode) {$e_1$};
        \draw[-{Latex[length=2mm]}] (2) -- (5) node [pos=.35, above, sloped] (TextNode) {$e_3$};

        \draw[-{Latex[length=2mm]}] (3) -- (6) node [midway, above, sloped] (TextNode) {$e_2$};
        \draw[-{Latex[length=2mm]}] (4_5) -- (7) node [pos=.3, above, sloped] (TextNode) {$e_4$};
        \draw[-{Latex[length=2mm]}] (6_7) -- (8) node [midway, above, sloped] (TextNode) {$e_5$};
    \end{tikzpicture}
    }
\end{center}
    
    \caption{Original metagraph.}
    \end{subfigure}

    \begin{subfigure}{0.95\linewidth}
     \begin{center}
     
\resizebox{0.7\textwidth}{!}{
    \begin{tikzpicture}
        \node[minimum width=1cm,draw,circle] (1) at (0,3.5) {$\mathit{x_1}$};
        \node[minimum width=1cm,circle] (2) at (0,1){$\mathit{x_2}$};
        \node[ellipse, draw,minimum width = 2cm, minimum height = 4cm] (1_2) at (0,2.25){};
        \node[ellipse, draw, minimum width = 2cm, minimum height = 4cm, draw] (6_7) at (5,2.5){};

        \node[minimum width=1cm,draw,circle] (6) at (5,3.5) {$\mathit{x_6}$};
        \node[minimum width=1cm,circle] (7) at (5,1.5){$\mathit{x_7}$};
        \node[minimum width=1cm,draw,circle] (8) at (7.5,2.5){$\mathit{x_8}$};
        \node[ellipse, rotate around={21:(0,0)}, minimum width = 4.85cm, minimum height = 2cm, draw] (7_8) at (6.25,2){};

        \path[-{Latex[length=2mm]}] (1) edge [bend left] node [style={pos=.4,sloped,anchor=south,auto=false}] (TextNode) {$e'_1$}(6);
        \path[-{Latex[length=2mm]}] (1_2) edge[bend left = 5] node [style={sloped,anchor=south,auto=false}] (TextNode) {$e'_2$} (7_8);
        \path[-{Latex[length=2mm]}] (6_7) edge[bend left = 50] node [style={sloped,anchor=south,auto=false}] (TextNode) {$e'_3$} (8);
    \end{tikzpicture}
}
\end{center}
    \caption{Reported application of BPP with $X'=\{x_1, x_2, x_6, x_7, x_8\}$ to (a).}
    \label{fig:badprojection}
    \end{subfigure}

    \begin{subfigure}{0.95\linewidth}
     \begin{center}
     
\resizebox{0.7\textwidth}{!}{
    
    \begin{tikzpicture}
        \node[minimum width=1cm,draw,circle] (1) at (0,3.5) {$\mathit{x_1}$};
        \node[minimum width=1cm,circle] (2) at (0,1){$\mathit{x_2}$};
        \node[ellipse, draw,minimum width = 2cm, minimum height = 4cm] (1_2) at (0,2.25){};
        \node[ellipse, draw, minimum width = 2cm, minimum height = 4cm, draw] (6_7) at (5,2.5){};

        \node[minimum width=1cm,draw,circle] (6) at (5,3.5) {$\mathit{x_6}$};
        \node[minimum width=1cm,circle] (7) at (5,1.5){$\mathit{x_7}$};
        \node[minimum width=1cm,draw,circle] (8) at (7.5,2.5){$\mathit{x_8}$};
        \node[ellipse, rotate around={21:(0,0)}, minimum width = 4.85cm, minimum height = 2cm, draw] (7_8) at (6.25,2){};
        \node[ellipse, rotate around={158:(0,0)}, minimum width = 7.28cm, minimum height = 2.2cm, draw] (1_7) at (2.6,2.44){};

        \path[-{Latex[length=2mm]}] (1) edge [bend left] node [style={pos=.4,sloped,anchor=south,auto=false}] (TextNode) {$e'_1$}(6);
        \path[-{Latex[length=2mm]}] (1_2) edge[bend left=5] node [style={sloped,anchor=south,auto=false}] (TextNode) {$e'_2$} (7_8);
        \path[-{Latex[length=2mm]}] (6_7) edge[bend left = 50] node [style={sloped,anchor=south,auto=false}] (TextNode) {$e'_3$} (8);
        \path[-{Latex[length=2mm]}] (1_7) edge[bend right= 22 ] node [style={sloped,anchor=south,auto=false}] (TextNode) {$e'_4$} (8);

    \end{tikzpicture}
    }
\end{center}
    \caption{Actual application  of BPP with $X'=\{x_1, x_2, x_6, x_7, x_8\}$ to (a).}
    \label{fig:projection}
    \end{subfigure}

    \caption{Example isomorphic to that in \citet{blanning1997,basu2007metagraphs} showing their incorrect calculation of BPP and the corrected version derived here.}
    \label{fig:motivation}
\end{figure}
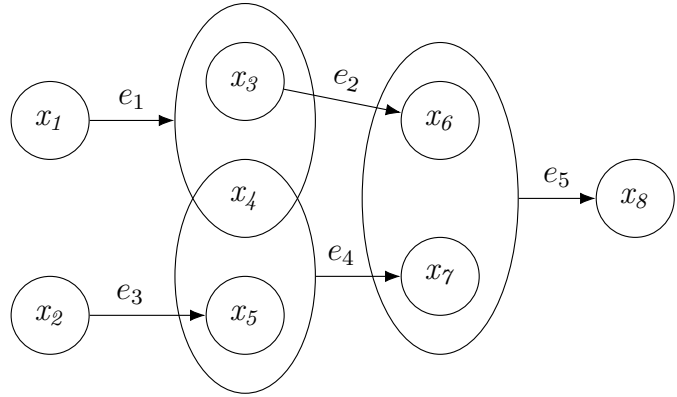
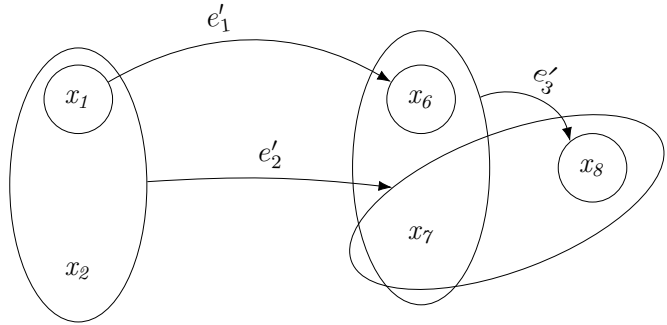
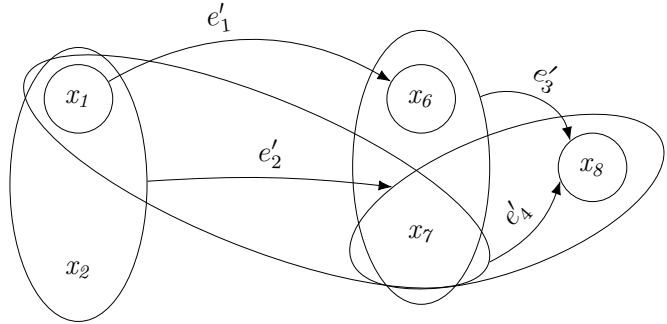


\autoref{fig:motivation} shows the metagraph $\mathit{\mathcal{H} = \langle X, E\rangle}$ with generating set $\mathit{X=\{ x_1, x_2, x_3, x_4, x_5, x_6, x_7, x_8\}}$ \sloppy and edges $\mathit{E = \{e_1, e_2, e_3, e_4, e_5\}}$. In \citep{basu2007metagraphs} and \citep{blanning1997} the authors calculate a projection of this metagraph over ${X'=\{x_1, x_2, x_6, x_7, x_8\}}$, the result of which is shown in \autoref{fig:badprojection}. In both pieces of work the authors provide the same projection (for completeness it is provided in \autoref{fig:badprojection}). Unfortunately in both cases the solution has an error which we will correct in a worked example.

On first inspection the projection in \autoref{fig:badprojection} looks simple enough to convince the reader the methodology will consistently produce projections that are simpler than the projected metagraph.

\subsection{Worked example}\label{sec:worked_example}
Our analysis begins by working through the example above so that we can begin with the correct projection for further investigations.

The construction of a projection revolves around the use of parts (2) and (3) of Definition \ref{def:projection} and is as follows:

\begin{enumerate}
    \item As per Definition \ref{def:projection} part (2), for every $\mathit{x'\in X'}$, for every $\mathit{V'\subseteq X'\backslash\{x\}}$ if there is a dominant metapath $\mathit{M(V',\{x'\})}$ in $\mathcal{H}$ then add it to a set of dominant metapaths $\mathit{\{M_1(V',\{x'\}),\ldots\}}$.
    \item For every unique invertex $\mathit{V'}$ of the metapaths in $\mathit{\{M_1(V',\{x'\}),\ldots\}}$ we construct an edge $\mathit{e' = \langle V', W' \rangle \in E'}$ where $\mathit{W' = \{\bigcup_{i}\,x_i |\exists\; M_1(V',\{x_i\}) \in \{M_1(V',\{x'\}),\ldots\}\}}$ thus ensuring Definition \ref{def:projection} part (3) is met.
\end{enumerate}

In \autoref{fig:motivation} the metagraph $\mathcal{H}=\langle\mathit{X,E}\rangle$ has generating set $\mathit{X = \{x_1, x_2, x_3, x_4, x_5, x_6, x_7, x_8\}}$ and edges $\mathit{E = \{e_1, e_2, e_3, e_4, e_5\}}$. Here we consider a projection over $\mathit{X'=\{x_1, x_2, x_6, x_7, x_8\}}$. Since there are no dominant metapaths from any subset of $\mathit{X'}$ to either $\mathit{x_1}$ or $\mathit{x_2}$ the set of dominant metapaths to be considered for the projection is as follows.

\begin{compactenum}
    \item $\mathit{M_1(\{x_1\}, \{x_6\}) = \{e_1,e_2 \}}$
    \item $\mathit{M_2(\{x_1, x_2\},\{x_7\}) = \{e_1,e_3,e_4 \}}$
    \item $\mathit{M_3(\{x_1, x_2\},\{x_8\}) = \{e_1,e_2,e_3,e_4,e_5 \}}$
    \item $\mathit{M_4(\{x_1, x_7\},\{x_8\}) = \{e_1,e_2,e_5 \}}$
    \item $\mathit{M_5(\{x_6, x_7\},\{x_8\}) = \{e_5\}}$
\end{compactenum}

The error in \autoref{fig:badprojection} stems from a failure to consider the dominant metapath between $\mathit{\{x_1, x_7\}}$ and $\mathit{\{x_8\}}$ using edges $\mathit{\{e_1, e_2, e_5\}}$.

Since $\mathit{M_2(\{x_1, x_2\},\{x_7\})}$ and $M_3(\mathit{\{x_1, x_2\},\{x_8\})}$ are the only dominant paths that share an invertex the metapaths are combined in the following way.

\begin{compactenum}
    \item $\mathit{C(e'_1) =\{ M_1(\{x_1\}, \{x_6\})\}}$
    \item $\mathit{C(e'_2) =\{M_2(\{x_1, x_2\},\{x_7\}), M_3(\{x_1, x_2\},\{x_8\})\}}$
    \item $\mathit{C(e'_3) =\{M_5(\{x_6, x_7\},\{x_8\})\}}$
    \item $\mathit{C(e'_4) =\{M_4(\{x_1, x_7\},\{x_8\})\}}$
\end{compactenum}

Leading to the following edges:

\begin{compactenum}
    \item $\mathit{e'_1 = (\{x_1\}, \{x_6\})}$
    \item $\mathit{e'_2 = (\{x_1, x_2\},\{x_7,x_8\})}$
    \item $\mathit{e'_3 =(\{x_6, x_7\},\{x_8\})}$
    \item $\mathit{e'_4 =(\{x_1, x_7\},\{x_8\})}$
\end{compactenum}

\FloatBarrier
Using the calculation above, \autoref{fig:projection} displays the corrected projection of the metagraph shown in \autoref{fig:motivation} with $\mathcal{H}=\langle\mathit{X,E}\rangle$, where $\mathit{X = \{x_1, x_2, x_3, x_4, x_5, x_6, x_7, x_8\}}$ and $\mathit{E = \{e_1, e_2, e_3, e_4, e_5\}}$, projected over $\mathit{X'=\{x_1, x_2, x_6, x_7, x_8\}}$.


Upon initial examination, the new visualization appears to be as complex as the metagraph shown in \autoref{fig:motivation}, from which it is projected. However, the discovery of edge $\mathit{e'_4 =(\{x_1, x_7\},\{x_8\})}$ uncovers the presence of two input sets that have a path to $\mathit{x_8}$. This realization led to the development of a scheme that illustrates how this projection method can generate more edges in the projection than were present in the original metagraph.

\subsection{Metagraph projection complexity}
\label{sec:mpcomplexity}
In \citet{ausiello2016} the authors define a succinct transitive closure for use in the context of B-Hypergraphs (equivalent to metagraphs with edges that have only a single element in the outvertex). It is defined as follows:

\begin{definition}[Succinct transitive closure of a hypergraph]
    \label{def:succinctclosure}
    Given a directed hypergraph $\mathit{\mathcal{H}= \langle N, H \rangle}$ the succinct transitive closure of $\mathcal{H}$ can be compactly represented by means of a directed hypergraph $\mathit{\mathcal{H}^+= \langle N, H^+ \rangle}$ where $\mathit{\mathcal{H}^+}$ contains a hyperarc $\mathit{(S, t)}$ if and only if $\mathit{S}$ is a source set in $\mathit{\mathcal{H}}$ and a hyperpath $\mathit{\langle S, t \rangle}$ exists in $\mathit{\mathcal{H}}$.
\end{definition}

Definition~\ref{def:succinctclosure} differs from Definition~\ref{def:projection} in that:
\begin{enumerate}
    \item There is no restriction of the closure to a subset of the vertices in the hypergraph.
    \item It does not enforce uniqueness on the invertex of the edges in $\mathit{\mathcal{H}^+}$.
    \item There is no requirement that hyperpaths forming the preimage of edges in $\mathit{\mathcal{H}^+}$ be dominant.
\end{enumerate}

Despite the different intended purposes of the two definitions, the similarities between these definitions and the fact that metagraph projection construction entails finding all possible dominant paths to each of the elements in the generating set leads us to the conclusion that metagraph projection is a form of transitive closure. Thus, we expect that in the general case the number of edges in $\mathcal{H'}$ to have a quadratic relationship with the number of edges in $\mathcal{H}$.

This is not ideal because the purpose of a metagraph projection is to provide a simplified view for visualization. However, some form of a minimal equivalent graph is necessary to capture the relationships between elements of $\mathit{X'}$ resulting from edges involving elements in $\mathit{X\backslash X'}$. Minimum equivalent B-Hypergraphs are discussed in \citet{ausiello1986} and address minimality in terms of redundancy in edges or sources, but that discussion is not relevant here as $\mathit{X'}$ determines the set of sources. Thus, the choice made by \citet{basu2007metagraphs} to apply a form of transitive closure to the relationships between elements of $\mathit{X'}$ seems reasonable as it does not introduce any additional vertex sets. However, their method can generate extra edges without any benefit as we demonstrate in Lemma \ref{complex}, where we show the negative consequences of this blanket approach.

Any measure of complexity should consider the number of edges in the projection, especially when it exceeds the number in the original metagraph.

\begin{lemma}
\label{complex}
    We can generate an infinitely countable set of graphs which have quadratic growth in their edges because the relationships between vertices in $\mathit{X}$ that could be deduced by transitivity are replaced with an edge.
\end{lemma}

\begin{proof}
    We start by constructing a metagraph from fragments of the metagraph in \autoref{fig:motivation}. The metagraph $\mathit{\mathcal{H}_n}$ is constructed from $n$ copies of the fragment shown in \autoref{fig:fragments}(\ref{sub@fig:fragments_body}) and one copy of the tail fragment shown in \autoref{fig:fragments}(\ref{sub@fig:fragments_tail}). The fragments are arranged such that body fragment $i=1$ connects to the tail fragment with edge $\mathit{e_{2,1}}$ connecting to element $\mathit{A_0}$ and edge $\mathit{e_{4,1}}$ connecting to element $\mathit{B_0}$ in the tail fragment. Then if $n > 1$ body fragment $i$ connects to body fragment $i-1$ with edge $\mathit{e_{2,i}}$ connecting to element $\mathit{A_{i-1}}$ and edge $\mathit{e_{4,i}}$ connecting to element $\mathit{B_{i-1}}$ in the $(i-1)$th body fragment. Thus, $\mathcal{H}_1$ is the example from \citet{basu2007metagraphs} shown in \autoref{fig:motivation} and $\mathcal{H}_2$ is shown in \autoref{fig:H_2}. Now consider a projection over $\mathit{X'_n = \{A_i, B_i \mid 0\leq i \leq n\}\bigcup F_0}$.
    Using the knowledge we gained from our worked example we can count the edges in the projection $\mathit{H'_n}$ of $\mathit{\mathcal{H}_n}$ over $\mathit{X'_n}$.
    There are three sets of edges to consider:
    \begin{compactenum}
        \item There is an edge for each of the pairs of inputs contained in $\mathit{X'_n}$ that have a path to $\mathit{F_0}$. Each input $\mathit{A_i}$ can be paired with inputs $\{\mathit{B_j \mid 0 \leq j \leq i \}}$ there are a total of $\mathit{\sum_{i=1}^{n}i+1}$ such edges.

        \item There is an edge for each input $\mathit{A_i}$ that has a path to $\mathit{A_0}$ there are $\mathit{n}$ such edges.

        \item The projection will contain one edge due to the edge in the tail fragment of $\mathcal{H}_n$.
    \end{compactenum}

    This gives a total of $\mathit{1 + n + \sum_{i=1}^{n}i+1 = (n^2+5n+2)/2}$ edges in the projection $\mathit{H'_n}$. The graph $\mathcal{H}_n$ has 4n + 1 edges. Thus, for n = 3 both $\mathit{H'_n}$ and $\mathit{H_n}$ have 13 edges and for $n > 3$ then $|E_n| < |E'_n|$.
\end{proof}

\FloatBarrier

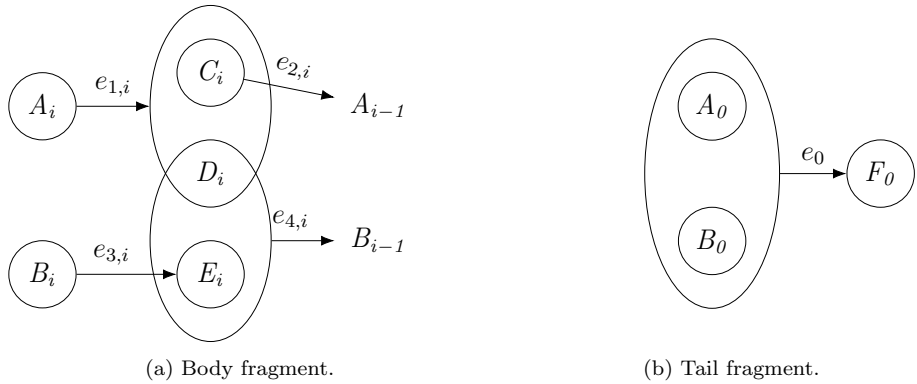
\begin{figure}[ht]
    \begin{center}
        \begin{subfigure}[b]{64mm}

            \resizebox{0.9\linewidth}{!}{

                \begin{tikzpicture}[baseline = {(0,0)}]
                    \node[minimum width=1cm,draw,circle] (A) at (0,3.5) {$\mathit{A_i}$};
                    \node[minimum width=1cm,draw,circle] (B) at (0,1){$\mathit{B_i}$};
                    \node[minimum width=1cm,draw,circle] (C) at (2.5,4) {$\mathit{C_i}$};
                    \node (D) at (2.5,2.5) {$\mathit{D_i}$};
                    \node[minimum width=1cm,draw,circle] (E) at (2.5,1){$\mathit{E_i}$};
                    \node[ellipse, minimum width = 1.8cm, minimum height = 3cm, draw] (C_D) at (2.5,3.5){};
                    \node[ellipse, minimum width = 1.8cm, minimum height = 3cm, draw] (D_E) at (2.5,1.5){};
                    \node[minimum width=1cm,circle] (F) at (5,3.5) {$\mathit{A_{i-1}}$};
                    \node[minimum width=1cm,circle] (G) at (5,1.5){$\mathit{B_{i-1}}$};
                    \draw[-{Latex[length=2mm]}] (A) -- (C_D) node [midway, above, sloped] (TextNode) {$e_{1,i}$};
                    \draw[-{Latex[length=2mm]}] (B) -- (E) node [pos=.35, above, sloped] (TextNode) {$e_{3,i}$};
                    \draw[-{Latex[length=2mm]}] (C) -- (F) node [midway, above, sloped] (TextNode) {$e_{2,i}$};
                    \draw[-{Latex[length=2mm]}] (D_E) -- (G) node [pos=.3, above, sloped] (TextNode) {$e_{4,i}$};
                \end{tikzpicture}
            }

            \caption{Body fragment.}
            \label{fig:fragments_body}
        \end{subfigure}
        \begin{subfigure}[b]{64mm}
            \resizebox{0.9\linewidth}{!}{
                \hspace{20mm}
                \begin{tikzpicture}[baseline={(0,0)}]
                    \node[minimum width=1cm,draw,circle] (A_0) at (0,3.5) {$\mathit{A_0}$};
                    \node[minimum width=1cm,draw,circle] (B_0) at (0,1.5){$\mathit{B_0}$};
                    \node[ellipse, minimum width = 2cm, minimum height = 4cm, draw] (A0_B0) at (0,2.5){};
                    \node[minimum width=1cm,draw,circle] (F) at (2.5,2.5){$\mathit{F_0}$};

                    \draw[-{Latex[length=2mm]}] (A0_B0) -- (F) node [midway, above, sloped] (TextNode) {$e_0$};
                \end{tikzpicture}
            }
            \caption{Tail fragment.}
            \label{fig:fragments_tail}
        \end{subfigure}

        \caption{Graph fragments for constructing $\mathcal{H}_n$.}
        \label{fig:fragments}
    \end{center}
\end{figure}
\FloatBarrier
As anticipated we have uncovered a quadratic relationship between $|E_n|$ and $|E'_n|$. To give a feel for the level of complexity that is reached quickly using this construction we include the graph $\mathit{\mathcal{H}_2 = \langle X_2, E_2 \rangle}$ where $X_2 = \{F_0, A_0, A_1, A_2, B_0, B_1, B_2, C_1, C_2, D_1, D_2, E_1, E_2 \} $ and $E_2 = \{e_{1,2}, e_{2,2}, e_{3,2}, e_{4,2}, e_{1,1}, e_{2,1}, e_{3,1}, e_{4,1}, e_0 \}$ in \autoref{fig:H_2} and $\mathit{H'_2}$ its projection over $X'_2 = \{F_0, A_0, A_1, A_2, B_0, B_1, B_2 \}$ in \autoref{fig:H_2p}. There are $9$ edges in $\mathcal{H}_2$ and only $8 = ((2)^2 + 5(2) + 2)/2$ edges in its projection $\mathcal{H'}_2$. Even with $n=2$ it is already apparent that the projection is more complex than the metagraph it should provide a simplified view of.

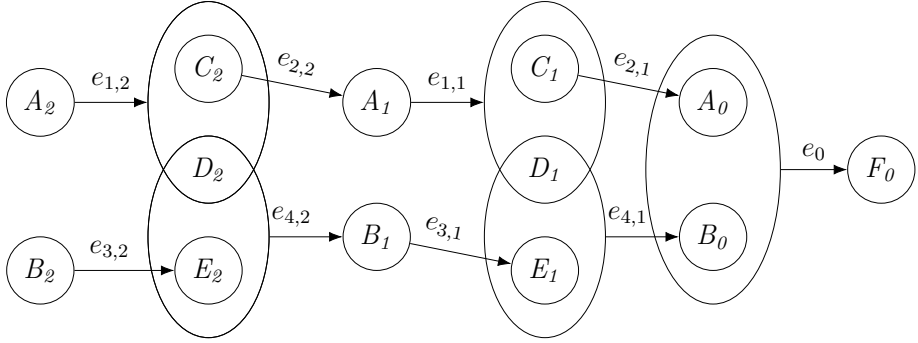
\begin{figure}[!htb]
\begin{center}
     
\resizebox{0.9\textwidth}{!}{
    \begin{tikzpicture}
        \node[minimum width=1cm,draw,circle] (1) at (0,3.5) {$\mathit{A_2}$};
        \node[minimum width=1cm,draw,circle] (2) at (0,1){$\mathit{B_2}$};

        \node[minimum width=1cm,draw,circle] (3) at (2.5,4) {$\mathit{C_2}$};
        \node (4) at (2.5,2.5) {$\mathit{D_2}$};
        \node[minimum width=1cm,draw,circle] (5) at (2.5,1){$\mathit{E_2}$};
        \node[ellipse, minimum width = 1.8cm, minimum height = 3cm, draw] (9) at (2.5,3.5){};
        \node[ellipse, minimum width = 1.8cm, minimum height = 3cm, draw] (10) at (2.5,1.5){};
        \node[minimum width=1cm,draw,circle] (6) at (5,3.5) {$\mathit{A_1}$};
        \node[minimum width=1cm,draw,circle] (7) at (5,1.5){$\mathit{B_1}$};

        \node[minimum width=1cm,draw,circle] (8) at (7.5,4) {$\mathit{C_1}$};
        \node (9) at (7.5,2.5) {$\mathit{D_1}$};
        \node[minimum width=1cm,draw,circle] (10) at (7.5,1){$\mathit{E_1}$};

        \node[minimum width=1cm,draw,circle] (11) at (10,3.5) {$\mathit{A_0}$};
        \node[minimum width=1cm,draw,circle] (12) at (10,1.5){$\mathit{B_0}$};
        \node[minimum width=1cm,draw,circle] (13) at (12.5,2.5){$\mathit{F_0}$};

        \node[ellipse, minimum width = 1.8cm, minimum height = 3cm, draw] (34) at (2.5,3.5){};
        \node[ellipse, minimum width = 1.8cm, minimum height = 3cm, draw] (45) at (2.5,1.5){};
        \node[ellipse, minimum width = 2cm, minimum height = 4cm, draw] (1112) at (10,2.5){};

        \node[ellipse, minimum width = 1.8cm, minimum height = 3cm, draw] (89) at (7.5,3.5){};
        \node[ellipse, minimum width = 1.8cm, minimum height = 3cm, draw] (910) at (7.5,1.5){};

        \draw[-{Latex[length=2mm]}] (1) -- (34) node [midway, above, sloped] (TextNode) {$e_{1,2}$};
        \draw[-{Latex[length=2mm]}] (2) -- (5) node [pos=.35, above, sloped] (TextNode) {$e_{3,2}$};

        \draw[-{Latex[length=2mm]}] (3) -- (6) node [midway, above, sloped] (TextNode) {$e_{2,2}$};
        \draw[-{Latex[length=2mm]}] (45) -- (7) node [pos=.3, above, sloped] (TextNode) {$e_{4,2}$};

        \draw[-{Latex[length=2mm]}] (6) -- (89) node [midway, above, sloped] (TextNode) {$e_{1,1}$};
        \draw[-{Latex[length=2mm]}] (7) -- (10) node [pos=.3, above, sloped] (TextNode) {$e_{3,1}$};

        \draw[-{Latex[length=2mm]}] (8) -- (11) node [midway, above, sloped] (TextNode) {$e_{2,1}$};
        \draw[-{Latex[length=2mm]}] (910) -- (12) node [pos=.3, above, sloped] (TextNode) {$e_{4,1}$};

        \draw[-{Latex[length=2mm]}] (1112) -- (13) node [midway, above, sloped] (TextNode) {$e_0$};
    \end{tikzpicture}
    }

\end{center}
    
    \caption{$\mathcal{H}_2$ from \autoref{complex}.}
    \label{fig:H_2}
\end{figure}

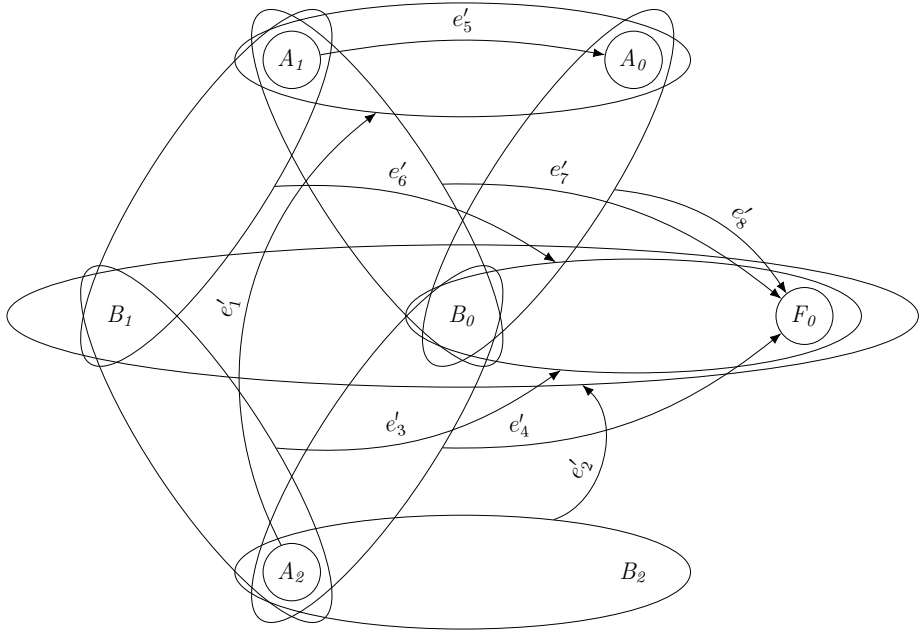
\begin{figure}[!htb]
 \begin{center}
     
\resizebox{0.9\textwidth}{!}{

    \begin{tikzpicture}

        \node[minimum width=1cm,draw,circle] (1) at (3,0) {$\mathit{A_2}$};
        \node[minimum width=1cm,draw,circle] (6) at (3,9) {$\mathit{A_1}$};
        \node[minimum width=1cm,draw,circle] (11) at (9,9) {$\mathit{A_0}$};
        \node[minimum width=1cm,draw,circle] (13) at (12,4.5) {$\mathit{F_0}$};
        \node[minimum width=1cm,circle] (12) at (6,4.5) {$\mathit{B_0}$};
        \node[minimum width=1cm,circle] (2) at (9,0) {$\mathit{B_2}$};
        \node[minimum width=1cm,circle] (7) at (0,4.5) {$\mathit{B_1}$};

        \node[ellipse, minimum width = 8cm, minimum height = 2cm, draw] (B) at (9,4.5){};
        \node[ellipse, minimum width = 8cm, minimum height = 2cm, draw] (C) at (6,0){};
        \node[ellipse, minimum width = 8cm, minimum height = 2cm, draw] (D) at (6,9){};
        \node[ellipse, rotate around={56.5:(0,0)}, minimum width = 7.4cm, minimum height = 2cm, draw] (E) at (1.5,6.75){};
        \node[ellipse, rotate around={56.5:(0,0)}, minimum width = 7.4cm, minimum height = 2cm, draw] (F) at (7.5,6.75){};
        \node[ellipse, rotate around={56.5:(0,0)}, minimum width = 7.4cm, minimum height = 2cm, draw] (G) at (4.5,2.25){};
        \node[ellipse, rotate around={123.5:(0,0)}, minimum width = 7.4cm, minimum height = 2cm, draw] (H) at (4.5,6.75){};
        \node[ellipse, rotate around={123.5:(0,0)}, minimum width = 7.4cm, minimum height = 2cm, draw] (I) at (1.5,2.25){};
        \node[ellipse, minimum width = 16cm, minimum height = 2.5cm, draw] (J) at (6,4.5){};

        \path[-{Latex[length=2mm]}] (1) edge [bend left =40] node [style={pos=.5,sloped,anchor=south,auto=false}] (TextNode) {$e'_1$}(D);
        \path[-{Latex[length=2mm]}] (C) edge [bend right =60] node [style={pos=.4,sloped,anchor=south,auto=false}] (TextNode) {$e'_2$}(J);
        \path[-{Latex[length=2mm]}] (I) edge [bend right =20] node [style={pos=.4,sloped,anchor=south,auto=false}] (TextNode) {$e'_3$}(B);
        \path[-{Latex[length=2mm]}] (G) edge [bend right =20] node [style={pos=.2,sloped,anchor=south,auto=false}] (TextNode) {$e'_4$}(13);
        \path[-{Latex[length=2mm]}] (6) edge [bend left =10] node [style={pos=.5,sloped,anchor=south,auto=false}] (TextNode) {$e'_5$}(11);
        \path[-{Latex[length=2mm]}] (E) edge [bend left =18] node [style={pos=.4,sloped,anchor=south,auto=false}] (TextNode) {$e'_6$}(B);
        \path[-{Latex[length=2mm]}] (H) edge [bend left =20] node [style={pos=.3,sloped,anchor=south,auto=false}] (TextNode) {$e'_7$}(13);
        \path[-{Latex[length=2mm]}] (F) edge [bend left =25] node [style={pos=.6,sloped,anchor=south,auto=false}] (TextNode) {$e'_8$}(13);

    \end{tikzpicture}
    }
    \end{center}
    \caption{$\mathcal{H}'_2$ from \autoref{complex}.}
    \label{fig:H_2p}

\end{figure}

Basu and Blanning did not intend for metagraph projections to look more complex than the metagraphs they are meant to simplify.  Understanding why this is the case is simply a matter of recognizing that the metagraph projection removes all transitivity from the metagraph. \citet{basu2007metagraphs} say that \enquote{there can be no simple paths of any length between two elements unless there is an edge in the projection connecting them}, but do not deduce that the removal of any transitivity means all pairwise relationships between elements in the projection must be explicitly represented with an edge.

The function for extending paths at each stage of this closure is  not associative as claimed on page 22 of \citet{basu2007metagraphs}, because  $\gamma(R) = Trunc(Cat(\gamma(a_{ik})_n,\gamma(b_{kj})_m))$ is not associative. This does not cause a problem in the current implementation as it only ever extends simple-paths one edge at a time, but associativity cannot be assumed in any proof.

While some form of transitive closure is needed to fill the gaps in the metagraph projection caused by excluding certain generating elements, it may be possible to take a lighter approach for the rest of the projection. With this idea in mind, we introduce the transitivity preserving projection (TPP) in the next section.

\section{Transitivity Preserving Projection}
When utilizing graphs as models, it is implicit that the connections between vertices possess a transitive nature. This transitivity is a fundamental feature of graph structures that allows us to reason about them. However, the metagraph projection presented in \citep{basu2007metagraphs} flattens all pairwise relationships between vertices in the metagraph into a single edge for each vertex pair. This flattening process can obscure relationships that rely on the transitivity property, such as determining if paths share a vertex. The consequence of this 'flattening' is the quadratic growth in edges shown in \autoref{complex}. To ensure clarity and minimize the number of edges, denoted as $\mathit{|E'|}$, we maintain transitivity within elements of $\mathit{X'}$, avoiding explicit pairwise relationships that can be inferred through transitivity. However, it is necessary to incorporate a form of transitive closure or minimum equivalent metagraph to capture the relationships involving elements of $\mathit{X'}$ resulting from edges between elements in $\mathit{X\backslash X'}$. 


Our replacement algorithms are based on a direct search of the power set of the edges of a graph for metapaths using a breadth first approach. 

\subsection{Transitivity preserving projection.}
The following definition is our work and offers better time complexity than Definition~\ref{def:projection}.  There is discussion of its properties elsewhere in our work.
\begin{definition}[Transitivity
        preserving projection (TPP)]
    \label{def:tpp}
    Given a metagraph $\mathcal{H}=\langle\mathit{X,E}\rangle$ and $\mathit{X'}\subseteq\mathit{X}$, a metagraph $\mathcal{H'}=\langle\mathit{X',E'}\rangle$ is a transitivity preserving projection (TPP) of $\mathcal{H}$ over $\mathit{X'}$ if:
    \begin{enumerate}
        \item For any $\mathit{e'}=\langle\mathit{V',W'}\rangle\in\mathit{E'}$ and for any $\mathit{x'}\in\mathit{W'}$ there is a dominant metapath $\mathit{M(V',\{x'\})}$ in $\mathcal{H}$ with the property that no proper subset of the edges of $\mathit{M(V',\{x'\})}$ form a dominant metapath $\mathit{M(U,\{x'\})}$ in $\mathcal{H}$ with $\mathit{U}\subseteq\mathit{X'}$ and
        \item For every $\mathit{x'\in X'}$, if there is any dominant metapath $\mathit{M(V,\{x'\})}$ in $\mathcal{H}$ with $\mathit{V}\subseteq\mathit{X'}$ and no proper subset of the edges of $\mathit{M(V,\{x'\})}$ form a dominant metapath $\mathit{M(U,\{x'\})}$ in $\mathcal{H}$ with $\mathit{U}\subseteq\mathit{X'}$, then there is an edge $\langle\mathit{V',W'}\rangle\in\mathit{E'}$ such that $\mathit{V'=V}$ and $\mathit{x'}\in\mathit{W'}$, and
        \item No two edges in $\mathit{E'}$ have the same invertex.
    \end{enumerate}

\end{definition}

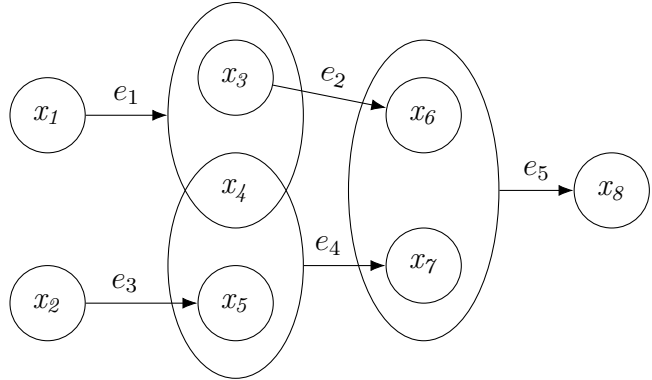
\begin{figure}[!htb]
    \centering
    \begin{tikzpicture}
        \node[minimum width=1cm,draw,circle] (1) at (0,3.5) {$\mathit{x_1}$};
        \node[minimum width=1cm,draw,circle] (2) at (0,1){$\mathit{x_2}$};

        \node[minimum width=1cm,draw,circle] (3) at (2.5,4) {$\mathit{x_3}$};
        \node (4) at (2.5,2.5) {$\mathit{x_4}$};
        \node[minimum width=1cm,draw,circle] (5) at (2.5,1){$\mathit{x_5}$};

        \node[ellipse, minimum width = 1.8cm, minimum height = 3cm, draw] (3_4) at (2.5,3.5){};

        \node[ellipse, minimum width = 1.8cm, minimum height = 3cm, draw] (4_5) at (2.5,1.5){};

        \node[minimum width=1cm,draw,circle] (6) at (5,3.5) {$\mathit{x_6}$};
        \node[minimum width=1cm,draw,circle] (7) at (5,1.5){$\mathit{x_7}$};
        \node[ellipse, minimum width = 2cm, minimum height = 4cm, draw] (6_7) at (5,2.5){};
        \node[minimum width=1cm,draw,circle] (8) at (7.5,2.5){$\mathit{x_8}$};
        \draw[-{Latex[length=2mm]}] (1) -- (3_4) node [midway, above, sloped] (TextNode) {$e_1$};
        \draw[-{Latex[length=2mm]}] (2) -- (5) node [pos=.35, above, sloped] (TextNode) {$e_3$};

        \draw[-{Latex[length=2mm]}] (3) -- (6) node [midway, above, sloped] (TextNode) {$e_2$};
        \draw[-{Latex[length=2mm]}] (4_5) -- (7) node [pos=.3, above, sloped] (TextNode) {$e_4$};
        \draw[-{Latex[length=2mm]}] (6_7) -- (8) node [midway, above, sloped] (TextNode) {$e_5$};
    \end{tikzpicture}
    \caption{Example isomorphic to that in \citep{blanning1997} and \citep{basu2007metagraphs}.}
    \label{fig:tppmotivation}
\end{figure}

We note that the TPP is the result of a filtering of the dominant paths used in the construction of the metagraph projection. It is this and the uniqueness of the metagraph projection which leads to the uniqueness of the TPP. However, for completeness we prove that uniqueness in \autoref{unique}.

\begin{theorem}\label{idempotent}
    The transitivity preserving projection is idempotent.
\end{theorem}

\begin{proof}
    Let $\mathcal{H'}=\langle\mathit{X',E'}\rangle$ be a projection of $\mathcal{H}=\langle\mathit{X,E}\rangle$ over $\mathit{X'}\subseteq\mathit{X}$ then each of the edges in $\mathcal{H'}$ are dominant by construction due to the fact that $\mathit{e'}=\langle\mathit{V',W'}\rangle\in\mathit{E'}$ is the reduction of a set of dominant metapaths $\mathit{M(V',\{x'\})}$ where $\mathit{x' \in W'}$.

    Now consider $\mathcal{H''}=\langle\mathit{X',E''}\rangle$ the transitivity preserving projection of $\mathcal{H'}=\langle\mathit{X',E'}\rangle$ over $\mathit{X'}$ then every dominant metapath $\mathit{M(V,\{x'\})}$ in $\mathcal{H''}$ with $\mathit{V}\subseteq\mathit{X'}$ that has 2 or more edges will have a proper subset of edges that form a dominant metapath $\mathit{M(U,\{x'\})}$ in $\mathcal{H'}$ with $\mathit{U}\subseteq\mathit{X'}$ (because every edge in $\mathcal{H'}$ is a dominant metapath in $\mathcal{H'}$). Thus, dominant metapaths with 2 or more edges are excluded from consideration in the transitivity preserving projection. Leaving only dominant edges in $\mathcal{H'}$  to be considered as metapaths in the projection (since all the edges in $\mathcal{H'}$ are dominant and satisfy the proper subset condition they are all represented in $\mathcal{H''}$).
\end{proof}

\begin{theorem}\label{forhung}
    The transitivity preserving projection is a projection in the mathematical sense.
\end{theorem}
\begin{proof}
    The transitivity preserving projection is a mapping from the set of possible directed edges between the vertexes selected from the power set of the generating set of $\mathcal{H'}$ to the set of possible directed edges between the vertexes selected from the power set of the generating set of $\mathcal{H'}$ (so its domain and codomain are the same mathematical structure) and by \autoref{idempotent} it is an idempotent mapping.
\end{proof}

\begin{theorem}\label{unique}
    The transitivity preserving projection of a metagraph over a subset of its generating set is unique and the composition of its edges is unique.
\end{theorem}
\begin{proof}
    By contradiction, let $\mathcal{H}=\langle\mathit{X,E}\rangle$ and $\mathcal{H'}=\langle\mathit{X',E'}\rangle$, $\mathcal{H''}=\langle\mathit{X',E''}\rangle$ be transitivity preserving projections of $\mathcal{H}$ over $\mathit{X'}$ with $\mathit{X'}\subseteq\mathit{X}$ and $\mathit{H' \neq H''}$.

    Without loss of generality since $\mathit{H' \neq H''}$ then $\mathit{\exists \: e' = \langle V', W' \rangle \in E'}$ such that $\mathit{e' \notin E''}$. Because $\mathcal{H'}$ is a transitivity preserving projection and $\mathit{e' \in E'}$ there exists a dominant metapath in $\mathcal{H}$ from $\mathit{V'}$ to $\mathit{W'}$ with the property that for each $\mathit{x \in W'}$ no proper subset of the edges of $\mathit{M(V',\{x'\})}$ form a dominant metapath $\mathit{M(U,\{x'\})}$ in $\mathcal{H}$ with $\mathit{U}\subseteq\mathit{X'}$.

    Since $\mathit{e' \notin E''}$ and $\mathcal{H''}$ is also a transitivity preserving projection of $\mathcal{H}$ then (by the definition of dominance) there must be an edge $\mathit{e'' = \langle V'', W'' \rangle \notin E''}$ with either $\mathit{V'\subset V''}$ and $\mathit{W'\subseteq W''}$ or $\mathit{V'\subseteq V''}$ and $\mathit{W'\subset W''}$. However, if $\mathit{V'\subset V''}$ then $\mathit{e''}$ violates condition 1 of Definition \ref{def:tpp} for all $\mathit{x \in W'}$. If $\mathit{V' = V''}$ and $\mathit{W'\subset W''}$ then $\mathit{e'}$ violates condition 3 of Definition \ref{def:tpp}. Thus, we must have $\mathit{e'=e''}$ a contradiction that proves the result.

    For any $\mathit{e' = \langle V', W' \rangle \in E'}$ then there must be at least one dominant metapath (with the correct subset condition on edges) $\mathit{M(V', W')}$ in $\mathcal{H}$ from $\mathit{V' \subseteq X'}$ to $\mathit{W' \subseteq X'}$. Now consider the set of such dominant metapaths $\mathit{\{M_1(V', W'), M_2(V', W'), M_3(V', W'), \ldots \}}$. Since there can only be one such set then $\mathit{C(e')}$ is unique and for each $\mathit{e'}$, $\mathit{C(e')}$ is unique.
\end{proof}

Here, we define the property of irreducibility for a metapath and establish its relationship to the Transitivity Preserving Projection (TPP). In  \autoref{thm:tpp_edges_dom_irred}, we prove that every edge in the TPP corresponds to a dominant and irreducible metapath in \(\mathcal{H}\), and conversely, every dominant and irreducible metapath in \(\mathcal{H}\) is represented as an edge in the TPP. Subsequently, in  \autoref{thm:factorisation} , we demonstrate that any dominant metapath in \(\mathcal{H}\) can be factored into a sequence of dominant irreducible factors, all of which are included in the TPP. This establishes that the TPP captures precisely all dominant metapaths in \(\mathcal{H}\).

\autoref{thm:TPP_minimality} then indicates that the TPP contains the minimum set of edges needed to preserve all dominant  relationships in \(\mathcal{H}\) restricted to \(X'\).


\begin{definition}[Composition of Paths (\(\circ\))]
\label{def:composition_simple}
Let \(\mathcal{H} = \langle X, E \rangle\) be a metagraph, where \(X\) is the generating set of vertices, 
and \(E\) is the set of edges. Let \(M(A, Z)\) and \(M(Z, \{x\})\) be two metapaths in \(\mathcal{H}\) such that 
the target set of \(M(A, Z)\) matches the source set of \(M(Z, \{x\})\), i.e., \(Z\).

The composition of \(M(A, Z)\) and \(M(Z, \{x\}), \) denoted as:
\[
   M(A, \{x\}) = M(A, Z) \circ M(Z, \{x\}), 
\]
is the path with an edge set defined as:
\[
   E(M(A, \{x\})) = E(M(A, Z)) \cup E(M(Z, \{x\})).
\]
\end{definition}

\begin{definition}[Factorization of a Path]
\label{def:factorization_simplified}
Let \(\mathcal{H} = \langle X, E \rangle\) be a metagraph, where \(X\) is the generating set of vertices, 
and \(E\) is the set of edges. A path \(M(A, \{x\})\) in 
\(\mathcal{H}\) is \emph{factorizable over \(X\)} if there exists a set \(Z \subseteq X\) such that:
\[
   M(A, \{x\}) = M(A, Z) \circ M(Z, \{x\}), 
\]
where:
\begin{enumerate}
    \item \(E(M(A, Z)) \cup E(M(Z, \{x\})) = E(M(A, \{x\}))\), 
    \item \(E(M(A, Z)) \cap E(M(Z, \{x\})) = \emptyset\), 
    \item \(Z\) consists of the \emph{pure inputs} required to produce \(x\), ensuring that 
    \(M(Z, \{x\})\) is minimal with respect to its inputs.
\end{enumerate}
\end{definition}

\begin{definition}[Irreducible Path over \(X'\)]
\label{def:irreducible_refined_final} Let \(\mathcal{H} = \langle X, E \rangle\) be a metagraph, where \(X\) is the generating set of vertices, 
\(E\) is the set of edges, and \(X' \subseteq X\) is a subset of the generating set.
A path \(M(A, \{x'\})\) with $A \subseteq X'\setminus \{x'\}, x' \in X'$, in a metagraph \(\mathcal{H}\) is \emph{irreducible over \(X'\)} if it cannot be factored further within \(X'\). Specifically, there does not exist a set \(Z \subseteq X'\) such that:
\[
   M(A, \{x\}) = M(A, Z) \circ M(Z, \{x\}), 
\]
where \(M(A, Z)\) and \(M(Z, \{x\})\) satisfy the conditions of a valid factorization over \(X\). Note that it is possible that a path can be factored over \(X\) but not over \(X'\) because we must have \(Z \subseteq X'\) and that \( M(Z, \{x\})\) might not be unique.
\end{definition}

\autoref{fig:not_unique} is an example of a metagraph that leads to a factorization which is not unique. Here we have two choices for \(Z\) either \(Z= \{E, D\}\) or \(Z= \{B, F\}\).

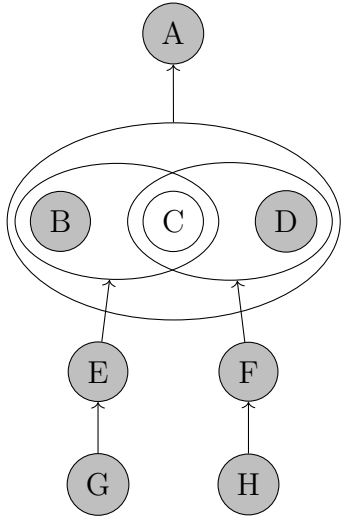
\begin{figure}[htbp]
   \centering
   \begin{tikzpicture}[
            scale = 1, transform shape,
   node distance =  11mm,
      res/.style = {ellipse, draw, minimum height=0.5cm, minimum width=0.8cm,
                     font=\footnotesize },
   literal/.style = {rectangle, draw, minimum height=0.5cm, text width=1.2 cm,
                     align=center, font=\footnotesize},
   FIT/.style args = {#1/#2/#3}{draw, ellipse, inner xsep=#2, fit=#3},
   every edge quotes/.style = { inner sep=1pt, font=\footnotesize},
                           ]
      \node[draw, circle, fill = lightgray] (G) at (-1,0) {G};
      \node[draw, circle, fill=lightgray] (E) at (-1,1.5) {E};
      \node[draw, circle, fill=lightgray] (H) at (1,0)  {H};
      \node[draw, circle, fill=lightgray] (F) at (1,1.5)  {F};

      \node[draw, circle, fill=lightgray] (B) at (-1.5,3.5) {B};
      \node[draw, circle] (C) at (0,3.5)    {C};
      \node[draw, circle, fill=lightgray] (D) at (1.5,3.5)  {D};

      \node[draw, circle, fill=lightgray] (A) at (0, 6) {A};

      \node[FIT=violet!30/-2mm/(B)(C)] (BC) {};

      \node[FIT=violet!30/-2mm/(C)(D)] (CD) {};
      \node[FIT=violet!30/-5.5mm/(BC)(CD)] (BCD) {};
      \draw[->] (G) -- (E);

      \draw[->] (H) -- (F);

      \draw[->] (E) -- (BC);

      \draw[->] (F) -- (CD);

      \draw[->] (BCD) -- (A);

   \end{tikzpicture}

   \caption{Metagraph that does not have a unique factorization over $X'={A, B, D, E, F, G, H}$.}
   \label{fig:not_unique}
\end{figure}
\begin{theorem}
\label{thm:tpp_edges_dom_irred}
Let $\mathcal{H} = \langle X, E\rangle$ be a metagraph, and let
$\mathcal{H'} = \langle X', E' \rangle$ be a \emph{transitivity preserving projection (TPP)} 
of $\mathcal{H}$ over $X' \subseteq X$. 

\textbf{(i)} If $M(A, \{x'\})$ in $\mathcal{H}$ is dominant and irreducible over $X'$, with $A \subseteq X'\setminus \{x'\}$ and $x' \in X'$, then there is an edge $\langle A, W'\rangle \in E'$ such that $x' \in W'$. In other words, every dominant, irreducible path in $\mathcal{H}$ with source and target in $X'$ is captured by a single edge in $\mathcal{H'}$.

\textbf{(ii)} If $\langle V', W'\rangle \in E'$ in $\mathcal{H'}$, then for each $x'' \in W'$, the path $M(V', \{x''\})$ in $\mathcal{H}$ is dominant and irreducible over $X'$. Hence each edge in $\mathcal{H'}$ corresponds exactly to a set of dominant, irreducible paths.
\end{theorem}

\begin{proof}
\textbf{(i)} Assume $M(A, \{x'\})$ is dominant and irreducible over $X'$, where $x' \in X'$ and $A \subseteq X'\setminus\{x'\}$. 
Dominance means $A$ is minimal for $x'$, and irreducibility over $X'$ means there is no factorization such that:
\[
   M(A, \{x'\}) = M(A, Z)\circ M(Z, \{x'\})
   \quad\text{with }Z\subseteq X'.
\]
By the second condition of the TPP definition, if there is a dominant path $M(V, \{y\})$ in $\mathcal{H}$ with $V \subseteq X'$ that cannot be further reduced in $X'$, then $\mathcal{H'}$ has an edge $\langle V, W''\rangle\in E'$ with $y\in W''$. Substituting $V=A$ and $y=x'$, there must be an edge $\langle A, W'\rangle\in E'$ for some $W'\subseteq X'$ satisfying $x' \in W'$. Therefore each dominant, irreducible path in $\mathcal{H}$ over $X'$ is captured by a single edge.

\textbf{(ii)} Conversely, let $\langle V', W'\rangle \in E'$. By the first condition of the TPP definition, for any $x'' \in W'$, there is a dominant path $M(V', \{x''\})$ in $\mathcal{H}$ no smaller subset of edges can replicate within $X'$. Thus $M(V', \{x''\})$ cannot be factored further in $X'$, so it is irreducible, and $V'\subseteq X'$ is minimal for $x''$, so it is dominant. Hence each edge $\langle V', W'\rangle$ corresponds exactly to dominant, irreducible paths $M(V', \{x''\})$ for $x''\in W'$.

Therefore $\mathcal{H'}$ contains exactly the edges representing dominant, irreducible paths from $\mathcal{H}$ restricted to $X'$, and no additional relationships appear in $E'$.
\end{proof}

\begin{lemma}
\label{lem:dominant_factors_formal}
Let $\mathcal{H} = \langle X, E \rangle$ be a metagraph with generating set $X$ and edge set $E.$ 
Suppose $M(A, \{x\})$ is a dominant metapath in $\mathcal{H}$ and can be factored as
\[
M(A, \{x\}) \;=\; M(A, Z) \;\circ\; M(Z, \{x\}).
\]
Then both $M(A, Z)$ and $M(Z, \{x\})$ are also dominant.
\end{lemma}

\begin{proof}
Consider the factorization $M(A, \{x\}) = M(A, Z) \circ M(Z, \{x\}).$ 
First, to show that $M(Z, \{x\})$ is dominant, assume for contradiction that it is not dominant. 
If $M(Z, \{x\})$ is not input dominant, there exists a proper subset $Z' \subset Z$ giving a valid 
metapath $M(Z', \{x\})$ whose edges are taken from $M(Z, \{x\}).$ Because $Z$ is the minimal set of 
inputs for producing $x$, replacing $Z$ by $Z'$ removes at least one edge from $M(Z, \{x\}), $ 
contradicting the edge dominance of $M(A, \{x\})$ since then $M(A, Z') \circ M(Z', \{x\})$ becomes 
valid but lacks that edge. If instead $M(Z, \{x\})$ is not edge dominant, there is some edge in 
$M(Z, x)$ that can be removed while still covering $\{x\}, $ again contradicting the edge dominance 
of $M(A, \{x\}).$ Hence $M(Z, \{x\})$ must be dominant.

A similar argument applies to $M(A, Z).$ Suppose $M(A, Z)$ is not dominant. If it is not input 
dominant, then there is a smaller subset $A' \subset A$ giving a valid path $M(A', Z).$ This 
contradicts the input dominance of $M(A, \{x\})$ because $M(A', Z) \circ M(Z, \{x\})$ would be 
a valid metapath using the same overall edge set but a strictly smaller source $A'.$ If $M(A, Z)$ 
is not edge dominant, one can remove an edge from $M(A, Z)$ without destroying coverage of $Z, $ 
which contradicts the edge dominance of $M(A, \{x\}).$ Thus $M(A, Z)$ is also dominant.
\end{proof}

\begin{theorem}
\label{thm:factorisation}
Let \(\mathcal{H} = \langle X, E\rangle\) be a metagraph, where \(X\) is the generating set of vertices, and let \(X' \subseteq X\). Let \(\mathcal{H'} = \langle X', E'\rangle\) be a Transitivity Preserving Projection (TPP) of \(\mathcal{H}\) over \(X'\). Suppose \(M(A,\{x\})\) in \(\mathcal{H}\) is a dominant metapath with \(A \subseteq X'\setminus\{x\},  x \in X'\). Then \(M(A,\{x\})\) can be decomposed into a finite set of dominant, irreducible sub-metapaths over \(X'\), each corresponding to an edge in \(\mathcal{H'}\). Consequently, every dominant metapath in \(\mathcal{H}\) whose support is contained in  X' is fully represented in \(\mathcal{H'}\) by its irreducible factors.
\end{theorem}

\begin{proof} 
By recursive decomposition.\\
\begin{enumerate}
    \item \textbf{Base Case.}
If \(M(A,\{x\})\) is irreducible over \(X'\), no factorization is required. By Condition~(2) of the TPP, there exists an edge \(\langle A, W'\rangle \in E'\) in \(\mathcal{H'}\) such that \(x \in W'\), representing \(M(A,\{x\})\). The process terminates for this path.

\item \textbf{Factorization.} 
If \(M(A,\{x\})\) is not irreducible over \(X'\), choose \(Z \subseteq X'\) such that
\[
   M(A,\{x\}) = M(A,Z) \circ M(Z,\{x\}),
\]
and \(M(Z,\{x\})\) is irreducible. By  \autoref{lem:dominant_factors_formal}, both \(M(A,Z)\) and \(M(Z,\{x\})\) remain dominant. Since \(M(Z,\{x\})\) is irreducible and dominant, Condition~(2) of the TPP ensures there is an edge \(\langle Z, W'\rangle \in E'\) with \(x \in W'\), representing \(M(Z,\{x\})\).

\item \textbf{Decomposition.}
Let \(Z = \{z_1, z_2, \dots, z_m\}\). For each \(z_i \in Z\), define
\begin{align*}
  \mathcal{M}(A, z_i)
  =
  \Bigl\{
    &M(A_j, \{z_i\})_k
    \;\Bigm|\;
    A_j \subseteq A,\ 
    M(A_j, \{z_i\})_k \text{ is dominant},\\
    &E\bigl(M(A_j, \{z_i\})_k\bigr)
    \subseteq
    E(M(A, Z))
  \Bigr\},
\end{align*}


where \(M(A_j, \{z_i\})_k\) denotes the \(k\)-th dominant path from the subset \(A_j \subseteq A\) to \(z_i\). Then
\[
  E\bigl(M(A,Z)\bigr)
  =
  \bigcup_{\,z_i \in Z}
    \;\bigcup_{\,M(A_j, \{z_i\})_k \,\in\, \mathcal{M}(A,z_i)}
      E\bigl(M(A_j,\{z_i\})_k\bigr).
\]

\item \textbf{Recursive Step.}  
For each \(M(A_j,\{z_i\})_k \in \mathcal{M}(A,z_i)\), recursively apply Steps~1–4 to factor \(M(A_j,\{z_i\})_k\). Each factorization step reduces the edge set, ensuring termination since \(X'\) is finite.
\end{enumerate}

Since each factorization step reduces the relevant edge set, and \(X'\) is finite, the recursion terminates in finitely many steps. Every branch of the decomposition ends in irreducible, dominant sub-metapaths, each represented by an edge in \(\mathcal{H'}\) via Condition~(2). Consequently, \(M(A,\{x\})\) decomposes into a finite set of irreducible components, ensuring that every dominant metapath in \(\mathcal{H}\) restricted to \(X'\) is captured by the TPP.
\end{proof}

\begin{theorem}[Minimality of the TPP]
\label{thm:TPP_minimality}
Let \(\mathcal{H} = \langle X, E \rangle\) be a metagraph, where \(X\) is the generating set of vertices, 
and let \(X' \subseteq X\). Let \(\mathcal{H'} = \langle X', E' \rangle\) be a 
Transitivity Preserving Projection (TPP) of \(\mathcal{H}\) over \(X'\). Suppose that there is another 
projection \(\mathcal{H''} = \langle X', E''\rangle\) of \(\mathcal{H}\) over \(X'\) that encodes 
all dominant metapaths whose source and target lie in \(X'\). Then \(\lvert E' \rvert \le \lvert E'' \rvert\). 
Consequently, \(\mathcal{H'}\) contains the minimum set of edges needed to  encode 
all dominant metapaths whose source and target lie in \(X'\). 
\end{theorem}

\begin{proof}
Suppose, for contradiction, that there exists another projection 
\(\mathcal{H''} = \langle X', E''\rangle\) of \(\mathcal{H}\) over \(X'\) 
such that \(\lvert E''\rvert < \lvert E'\rvert\) and \(\mathcal{H''}\) 
nevertheless encodes every dominant metapath with source and target in \(X'\).

From the definition of a TPP, each edge \(\langle V', W'\rangle\in E'\) 
corresponds to a metapath \(M(V', \{x'\})\) in \(\mathcal{H}\) that is both 
dominant and irreducible over \(X'\). No proper subset of its edges forms 
another dominant path to \(x'\) within \(X'\), so \(M(V', \{x'\})\) cannot 
be covered by any smaller edge set. Because \(\mathcal{H'}\) has strictly 
more edges than \(\mathcal{H''}\) by assumption, at least one edge of 
\(\mathcal{H'}\) is missing in \(\mathcal{H''}\). Let 
\(\langle U, W^*\rangle\in E'\setminus E''\). By the TPP property, 
\(\langle U, W^*\rangle\) represents a dominant, irreducible metapath 
\(M(U, \{y\})\) with \(y \in W^*\).

Because \(M(U, \{y\})\) is irreducible over \(X'\), no smaller set of edges 
produces a path from \(U\) to \(y\) while retaining dominance, so its removal 
discards coverage of at least one irreducible, dominant path. The projection 
\(\mathcal{H''}\), having fewer edges, omits \(\langle U, W^*\rangle\). 
Thus \(\mathcal{H''}\) fails to represent the dominant path \(M(U, \{y\})\). 
This contradicts the assumption that \(\mathcal{H''}\) encodes every dominant 
metapath among vertices in \(X'\).

Hence no projection with fewer edges than \(\mathcal{H'}\) can encode all 
dominant metapaths from \(\mathcal{H}\) restricted to \(X'\). Therefore 
\(\lvert E'\rvert \le \lvert E''\rvert\) for any other projection 
\(\mathcal{H''} = \langle X', E''\rangle\). This proves that 
\(\mathcal{H'}\) is minimal among all projections over \(X'\).
\end{proof}

\section{Transitivity preserving projection algorithms}

The outline of a search for a metapath between source set $s$ and a target set $t$ is simply a matter of setting $B=s$ and repeatedly calling the subset function of $S$ passing it the set $B$. At each call the outvertices of the edges returned are added to $B$ and the process continues until $t \subseteq B$. There are many details to consider to effectively find all the paths between $s$ and $t$ in this manner, but first we will consider a simpler case.


\subsection{Finding a single path.} \label{sec:fsp}
Frequently any path that connects source set $s$ and target set $t$ is sufficient, i.e.\ as proof of connectivity. Unfortunately the algorithms due to \citet{basu2007metagraphs} do not consider this simple case, every path is calculated using simple-paths and the closure $A^*$ of the adjacency matrix $A$. Using the direct approach discussed above will provide a path in $O(m^2)$ time (where $m$ is the number of edges in the metagraph). Its implementation is shown in  Algorithm~\ref{alg:GetSingleMetapathFrom}. As with all our algorithms it returns a list of references to edges that form the path.

 Between lines \ref{lst1:line:begin_forward} and \ref{lst1:line:end_forward} the algorithm starts at the $source$ and selects candidate edges to add to the $path$ on the basis that the candidate edge invertex set is in a set formed by the generating set elements in the source and the set of elements that make up the outvertices of the edges already added to the path. If a candidate edge is already in the path or the whole of its outvertex set is in the set made of the elements of the outvertices of the edges in the path it is rejected, otherwise the candidate edge is appended to the path.

 At this point of the algorithm \textit{path} contains a path from \textit{source} to \textit{target}, but it can contain many redundant edges. Thus, lines \ref{lst1:line:begin_backward} to \ref{lst1:line:end_backward} are concerned with removing the redundant edges. The process is to work backwards from the \textit{target} building a \textit{new\_path} using only those edges from \textit{path} that have outvertex sets that are not disjoint from a set \textit{required}, formed by starting with the \textit{target} and as each edge \textit{e} is added to the \textit{new\_path} by first deleting any elements in \textit{required} that are contained within \textit{outvertex}$(e)$ and then setting \textit{required} to be the union of \textit{required} and \textit{invertex}$(e)$. This ensures that ensuring every edge added to \textit{new\_path} is on a path that can reach some part of the \textit{target}.

 \begin{algorithm}[H]
     \caption{Find any metapath between a source set and a target set}\label{alg:GetSingleMetapathFrom}
     \begin{algorithmic}[1]
         \Function{GetSingleMetapathFrom}{E, source, target}
         \State $path \gets \text{empty list}$
         \State $B \gets \text{set containing elements of } \textit{source}$

         \State $\text{S} \gets \text{new SetTrieMultiMap}$
         \ForAll{$e_i \in E$}
         \State add $\text{invertex}(e_i)$ to $\text{S}$ with value $i$
         \EndFor

         \State $updated \gets \text{true}$

         \While{$updated$ and $target$ is not subset of $B$} \label{lst1:line:begin_forward}
         \State $updated \gets \text{false}$
         \State $I  \gets \text{list of } i \text{ mapped to subsets of } B  \text{ in }S$

         \ForAll{$i \in I$}
          \If{$i \in path \text{ or outvertex}(e_i) \subseteq B$}
         \State \textbf{continue}
         \EndIf

         \State $updated \gets \text{true}$
         \State $B \gets B \cup\text{outvertex}(e_i)$

         \State append $i$ to $path$
         \EndFor
         \EndWhile

         \If{$target \nsubseteq B$}
        \State \textbf{return} empty list \label{lst1:line:end_forward}
         \EndIf
         \State $new\_path \gets \text{empty list}$ \label{lst1:line:begin_backward}
         \State $required \gets \text{set containing elements of } \textit{target}$

         \ForAll{$i \in \text{reversed}(path)$}
         \If{$\text{outvertex}(e_i)\cap required \neq \emptyset$}
         \State append $i$ to $new\_path$
         \State $required \gets required \backslash \text{outvertex}(e_i)$
         \State $required \gets required \cup \text{invertex}(e_i)$
         \EndIf
         \EndFor

         \State reverse $new\_path$
         \State \textbf{return} $new\_path$ \label{lst1:line:end_backward}
         \EndFunction
     \end{algorithmic}
 \end{algorithm}
In order for an exhaustive search for metapaths to be efficient, we need to keep the search space as small as possible and in this case that means reducing the set of edges considered to a minimum. This presents us with a design decision. The algorithm for finding paths in \citet{basu2007metagraphs}  will find paths  between a given \textit{source} set of elements and a \textit{target} set of elements that exhibit the property that a proper subset of the edges can form a path between the source and the target. Note, however that such a path may dominate the path made from a subset of its edges because it may use  a smaller set of elements of the \textit{source} as its inputs.
The property can be thought of as the path being a walk that leaves the source and then revisits the source before arriving at the target.

\subsection{Minimizing the search space}

Our primary uses for finding all the metapaths in a metagraph are to analyze redundancy such as cut-set calculations and to calculate projections. By definition these paths do nothing to add to redundancy between  \textit{source} and \textit{target} sets, and we have created an improved form of projection in Definition \ref{def:tpp} which mandates filtering of these paths. Thus, the decision was made to exclude this type of path from our search. This enables an algorithmic simplification which is explained in Section~\ref{sec:fam}.  This design decision allows us to restrict the search space and Algorithm~\ref{alg:GetSuperpath} is the implementation that reduces the search space.

Between lines \ref{lst2:line:begin_forward} and \ref{lst2:line:end_forward} the algorithm starts at the $source$ and adds any edge the $path$ if its invertex set is contained in a set formed by the generating set elements in the \textit{source} and the set of elements that make up the outvertices of the edges already added to the path.

On completion of this part of the algorithm, \textit{path} contains every path that originates from the \textit{source} set and this may include paths that do not reach the \textit{target}. Thus, lines \ref{lst2:line:begin_backward} to \ref{lst2:line:end_backward} are concerned with removing edges that do not lie on any path between the \textit{source} and \textit{target}. The method is to work backwards from the \textit{target} building a \textit{new\_path} using only those edges from \textit{path} that have outvertex sets that are not disjoint from a set \textit{required}, formed by starting with the \textit{target} and as each edge \textit{e} is added to the \textit{new\_path} by setting \textit{required} to be the union of \textit{required} and \textit{invertex}$(e)$, thus ensuring every edge added to \textit{new\_path} is on a path that can reach some part of the \textit{target}.

\begin{algorithm}[H]
    \caption{Get all edges that lie on any path between source and target}\label{alg:GetSuperpath}
    \begin{algorithmic}[1]
        \Function{GetSuperpath}{E, source, target}

        \State $path \gets \text{empty list}$
        \State $B \gets \text{set containing elements of } \textit{source}$

        \State $\text{S} \gets \text{new SetTrieMultiMap}$
        \ForAll{$e_i \in E$}
        \State add $\text{invertex}(e_i)$ to $\text{S}$ with value $i$
        \EndFor

        \State $updated \gets \text{true}$
        \While{$updated$} \label{lst2:line:begin_forward}
        \State $updated \gets \text{false}$
        \State $I  \gets \text{list of } i \text{ mapped to subsets of } B  \text{ in }S$

        \ForAll{$i \in I$}
        \If{$i \in path$}
        \State \textbf{continue}
        \EndIf

        \State $B \gets B \cup \text{outvertex}(e_{\text{i}})$
        \State \text{append} $i$ \text{to} $path$
        \State $updated \gets \text{true}$
        \EndFor
        \EndWhile \label{lst2:line:end_forward}

        \State $new\_path \gets$ \text{empty list}
        \State $required \gets$ \text{set containing elements of} $target$
        \State $updated \gets \text{true}$
        \While{$updated$}
        \State $updated \gets \text{false}$
        \ForAll{$i \in \text{reversed}(path)$} \label{lst2:line:begin_backward}
        \If{$i \in new\_path$}
        \State \textbf{continue}
        \EndIf
        \If{$\text{outvertex}(e_i)\cap required \neq \emptyset$}
        \State \text{append} $i$ \text{to} $new\_path$
        \State $required \gets required \cup \text{invertex}(e_{i})$
        \State $updated \gets \text{true}$
        \EndIf
        \EndFor
        \EndWhile
        \State \textbf{return} $new\_path$ \label{lst2:line:end_backward}
        \EndFunction
    \end{algorithmic}
\end{algorithm}

\subsection{Finding every edge-dominant metapath between source and target sets.}\label{sec:fam}
There are three important observations that we use to minimize the work involved in finding all the edge dominant paths between \textit{source} and \textit{target} sets with the proper subset property discussed in Section~\ref{sec:fsp}.
\begin{enumerate}

    \item Using a breadth first approach means the first metapath discovered between \textit{source} and \textit{target} has the least number of edges such a metapath can have, and thus, has no redundant edges. In fact subsequently discovered metapaths with sets of edges that are not a superset of the edge set of any previously discovered metapaths are also edge redundant.  This is significant because as discussed in Section~\ref{sec:fsp}, we do not wish to include metapaths that exhibit the property that a proper subset of the edges can form a metapath between the \textit{source} and the \textit{target}. By using previously found metapaths to compare against subsequently found metapaths for this proper subset property we can exclude such metapaths.  This test also has the added benefit of excluding metapaths that have redundant edges. It is a lightweight test that can be implemented efficiently with a set-trie.

    \item The order in which edges are added to a metapath does not affect the set of outvertices available for extending the metapath. This observation yields two useful ideas we can leverage:
          \begin{enumerate}
              \item  As a breadth first search progresses we may visit the same set of edges in multiple ways, but if possible it is advantageous to avoid this. The usual technique employed in breadth first searches is recursion, but recursion does not allow for the search paths to be consolidated if their intermediate results are identical. Thus, an iterative technique that allows search paths to be consolidated is employed. It is also lightweight in that rather than the function call stack keeping track of the search (as in recursive functions) a dictionary does this.
              \item  The order in which edges are visited while searching for a metapath does not affect the set of outvertices available for its extension. Thus, the only state information needed to efficiently compute if an edge may extend the path is a set containing references to the edges in the path, and a set containing the union of the elements in the source set and elements in the outvertices of the edges already part of the path. This fits well as a dictionary can map sets of references to edges to the sets of elements that can be used to extend the path.
          \end{enumerate}

    \item If multiple simple edges leave an invertex it is sometimes possible to combine them and reduce the size of the search space. In the metagraph that  Figure~\ref{fig:aircraftmetagraph} is adapted from there are 14 edges leaving vertex \textit{9999}, leading to an immediate branching with $2^{14}$ possible choices at the start of our search in an example directed hypergraph model for an air-craft landing leg~\citep{hopp2023}. Obviously this has a large impact on the size of the search space. The idea is to replace edges like this with a single edge that has an outvertex that is the union of the outvertices of the edges it replaces. However, there are caveats to which edges this can be applied to. Indeed, those caveats apply even in the case discussed since there are two paths to vertex \textit{303}. The caveat is that we must not combine edges that would leave a path in the original metapath that could not be explored. In our case we use a very conservative approach to combining edges and only combine edges that share an invertex if all the elements in their outvertices only have one incoming edge.
\end{enumerate}

Algorithm~\ref{alg:getallpaths} uses each of these observations in a carefully choreographed manner to minimize the search space. We also have a practical consideration. Our testing of the algorithms in \citet{blanning1997} often resulted in running out of memory on small problems. In part this was due to holding many paths in memory simultaneously. Taking this into account as well as the fact that many algorithms that will use these paths perform a form of filtering that can be applied one path at a time we implemented the algorithm as a generator function so that paths can be consumed as they are generated rather than kept in memory until the search is exhausted.
\subsubsection{Edge combining.}
On line \ref{lst3:lin:superpath} the algorithm calls the function described in Algorithm~\ref{alg:GetSuperpath} which reduces the set of edges to be searched and references to these edges are now stored in $E_1$.

\vspace{10mm}
\begin{breakablealgorithm}
    \caption{Finding all the edge dominant paths between a source and target}\label{alg:getallpaths}
    \begin{algorithmic}[1]
        \Function{GetAllMetaPathsFrom}{$E_0$, source, target}
        \State $E_1 \gets \text{GetSuperpath}(E_0, source, target)$ \label{lst3:lin:superpath}

        \State $S_1 \gets$ \text{empty SetTrieMultiMap} \label{lst3:lin:startcontraction}

        \ForAll {$e_i \in E_1$}
        \State add $\text{invertex}(e_i)$ to $S_1$ with value $i$
        \EndFor

        \State $H \gets \text{list}(a_j)\text{ where } a_j\in \text{outvertext}(e_i) \text{ and }e_i \in E_0$

        \State $K \gets$ \text{empty SetTrie}
        \For{$h_i \in H$}
        \State add $h_i$ to $K$ if there is more than one copy of $h_i$ in $H$
        \EndFor

        \State $Q \gets \text{list}(\text{keys}(S_1))$
        \State $E_2 \gets$ \text{list}($E_0$)
        \State $S_2 \gets$ \text{empty SetTrieMultiMap}
        \State $D \gets$ \text{empty MultiDict}
        \ForAll {$I \in Q$}
        \State $L  \gets \text{list of edge indices mapped to } I  \text{ in }S_1$
        \State $O \gets \text{list}( \text{outvertex}(e_i) \text{ where } i \in L \text{ and }e_i \in E_2)$
        \State $n \gets \text{list}(o \in O \text{ where } o \text{ has no subsets in } K)$
        \If {$|L| < 2$ \text{or} $|n| < 2$} \label{lst3:lin:combine_edge_branch}
        \ForAll{$i \in L$}
        \State $\text{add }I \text{ to }S_2 \text{ with value }i$
        \EndFor
        \Else
        \State $S_0 \gets \text{empty list}$
        \State $S_0 \gets i \in L \text{ where } e_i \in E_2 \text{ and outvertex} (e_i) \notin n$

        \ForAll{$i \in S_0$}
        \State $\text{add }I \text{ to }S_2 \text{ with value }i$
        \EndFor

        \State $T \gets i \in L \text{ where } e_i \in E_2 \text{ and outvertex}(e_i) \in n$

        \If{$|T| > 0$}
        \State $out\gets \text{empty set}$
        \ForAll{$i \in T$}
        \State $out \gets out \cup \text{outvertex}(e_i) \text{ where } e_i \in E_2$
        \EndFor
        \State $\text{append Edge}(I, \text{out})\text{ to }  E_2$
        \State $\text{new\_edge\_pos} \gets |E_2| - 1$
        \State $\text{add }I \text{ to } S_2\text{ with value new\_edge\_pos}$

        \ForAll{$i \in T$}
        \State $\text{add new\_edge\_pos to }D \text{ with value }i$ \label{lst3:lin:endcontraction}

        \EndFor
        \EndIf
        \EndIf
        \EndFor

        \State $M \gets$ \text{new SetTrie}
        \State $\text{prev\_paths} \gets$ \text{empty dictionary}

        \State $edges  \gets \text{list of }i \text{ mapped to subsets of } source  \text{ in }S_2$
        \ForAll{$i \in edges$} \label{lst3:lin:add_source}
        \State $\text{add } \{i\} \text{ to prev\_paths with value } {source \cup \text{outvertex}(e_i)}\text{ where }e_i \in E_2$
        \EndFor

        \While{$|\text{prev\_paths}| > 0$} \label{lst3:line:search_start}
        \State $\text{del\_keys} \gets$ \text{empty list}

        \ForAll{$p \in$ $\text{prev\_paths}$}

        \If {$target \subseteq \text{value}(p) \text{ in }S_2$}
        \If{ $M$ contains no keys that are subsets of $p$}
        \State $\text{add }p\text{ to } M$
        \State \textbf{yield} $\text{ExpandMetapaths}(E_0, E_2, D, p)$
        \EndIf
        \State $\text{add }p \text{ to del\_keys}$
        \EndIf
        \EndFor

        \ForAll{$p \in\text{del\_keys}$}
        \State $\text{delete } p \text{ from prev\_paths}$
        \EndFor

        \State $\text{cur\_paths} \gets$ \text{empty dictionary}

        \ForAll{$p \in \text{prev\_paths}$}

        \State $edges  \gets \text{list of }i \text{ mapped to subsets of value} (p)  \text{ in }S_2$

        \ForAll{$i \in edges$}
        \If{$i \in p \text{ or outvertex}(e_i) \subseteq \text{value}(p) \text{ in }S_2 \text{ where } e_i \in E_2$}
        \State \textbf{continue}
        \EndIf

        \State $\text{add } p \cup \{i\} \text{ to cur\_paths with value } e_i \cup \text{value}(p)\text{ where }e_i \in E_2$
        \EndFor
        \EndFor

        \State $\text{prev\_paths} \gets \text{cur\_paths}$ \label{lst3:line:search_end}
        \EndWhile
        \EndFunction
    \end{algorithmic}
\end{breakablealgorithm}
\vspace*{5mm}
Then follows a great deal of detailed work as lines \ref{lst3:lin:startcontraction} to \ref{lst3:lin:endcontraction} implement edge combining. Here $H$ is a list of the generating set elements that feature in the outvertices of all the edges in the graph including repeats.
$K$ is a set-trie multi-map that contains only the elements of $H$ that have duplicates.
Thus, $K$ is a set-trie that contains elements found in vertices that have more than one incoming edge. Edges that visit these elements must not be combined. $S_1$ is a set-trie multi-map with keys that are the invertexes of the edges referenced in $E_1$ the set of edges the search is reduced to and $Q$ is simply a list of those invertices.

Having set the preconditions by calculating these pieces of information the algorithm iterates through each of the invertices in $Q$ with $I$ the invertex in question. By looking up $I$ in $S_1$ it is able to set $L$ to a list of references to edges that emanate from $I$. Using $L$ it is able to calculate $O$ which is a list of the outvertices of the edges in $L$ that emanate from $I$. Then the elements of $O$ are filtered to discard any outvertices which have elements in $K$ and the result is stored in $n$. Thus, $n$ stores the outvertices of the edges that emanate from invertex $I$ that are eligible for combining.

The set-trie multi-map $S_2$ is used to guide the search for paths in the same way as $S$ in Algorithm~\ref{alg:GetSingleMetapathFrom}. The lines from \ref{lst3:lin:combine_edge_branch} to \ref{lst3:lin:endcontraction}  ensure that $S_2$ can map invertexes to edge references for all the edges in our reduced search space $E_1$ and that $E_2$ is updated to include the new edges resulting from combining edges. Note that $E_2$ is a copy of $E_0$ so the references in $E_1$ which are actually indexes into the list $E_0$ are also suitable indexes into the list $E_2$.

The code branches in line \ref{lst3:lin:combine_edge_branch} based on the number of edges emanating from the vertex $I$ and the number of edges that are candidates for combining $|n|$. If either of these are less than two then there are insufficient edges to combine. Thus, the edges emanating from $I$ which are in the list $L$ are simply added to $S_2$ with $I$ as the key and $i\in L$ as the value.

If on line \ref{lst3:lin:combine_edge_branch} there are two or more edges emanating from $I$ and two or more edges eligible for combining $|n| \geq 2$ then the algorithm:
\begin{enumerate}
    \item Makes a list $S_0$ of the edges that will not be combined. This list is calculated by filtering the edges in $L$ keeping only those edges that do not have outvertices in $n$. These edges are added to $S_2$ with $I$ as the key and $i\in S_0$ as the value.
    \item  Makes a list $T$ of the edges that will  be combined. This list is calculated by filtering the edges in $L$ keeping only those edges that have outvertices in $n$. A new outvertex $out$ is formed by taking the union of outvertices of the edges in $T$. With the outvertex an edge from $I$ to $out$ is added to $E_2$ and an entry is added in $S_2$ with $I$ as the key and a reference to this new edge as the value which is stored in variable $new\_edge\_pos$. Finally, the multi-map $D$ which has the purpose of storing a map between a reference to the combined edge and the edges it is the result of combining is updated with the key being in $new\_edge\_pos$ and the values $i \in T$
\end{enumerate}

This detailed work is not essential for the main part of the algorithm to find paths but combining edges makes a large impact to runtimes on graphs with a lot of branching.

\subsubsection{Finding edge-dominant metapaths}
With edge combining complete the stage is set for the main work of the algorithm. The main variables and their interactions are as follows:

\begin{enumerate}
    \item There is an entry in the dictionary $prev\_paths$ for each path. The key is the set of edges that form the path and the value associated with that key is the set formed by taking the union of the $source$ set and all the outvertices of the edges in the path.
    \item The list of metagraph edges including those combined above are in $E_2$ and for every $e_i \in E$ the key-value $invertex(e_i)$ is mapped to $i$ in the set-trie multi-map $S_2$.  Thus given a $value(p)$ associated with a key $p$ in  $prev\_paths$ it is possible to efficiently find by reference all edges that have invertices wholly within $value(p)$ by simply calling the subset function of $S_2$ passing it the set $value(p)$.
    \item The set-trie $M$ has an entry for each metapath that the algorithm has found previously. The entry is identical to the key $p$ in $prev\_paths$ at the point that metapath $p$ was found. i.e.\ It is a set of references to the edges in the path that was yielded by the algorithm.
\end{enumerate}

For those unfamiliar with generator functions the yield call returns the value yielded but retains the state of the function, so subsequent calls to the function continue execution immediately after the call to yield.

The search is initiated with line \ref{lst3:lin:add_source} where all the edges $i$ that have outvertices that form subsets of $source$ are added as key values to the dictionary $prev\_paths$ the values associated with these keys are the union of the source set and the outvertex of $i$.

The algorithm then repeats the steps between lines \ref{lst3:line:search_start} and \ref{lst3:line:search_end} yielding when a suitable path is found but returning to the same loop until the dictionary $prev\_paths$ is empty. Those steps can be understood as follows:
\begin{enumerate}
    \item For every $p\in prev\_paths$ If the $target$ set is contained within the value associated with $p$ in the $prev\_paths$ dictionary then add $p$ to a list $del\_keys$ of entries that will be deleted from the dictionary $prev\_paths$ and check if $p$ is a superset of a previously found metapath in the set-trie $M$ if it is not then add it to $M$ and yield an expanded version of $p$ (this expansion is explained in the discussion of Algorithm \ref{alg:ExpandMetapaths}).
    \item Delete all the values stored in $del\_keys$ from $prev\_paths$.
    \item For every $p\in prev\_paths$ make a list $edges$ that have  invertices wholly within $\text{value}(p)$ by calling the subset function of $S_2$ passing it the set $\text{value}(p)$. Then for each $i$ in $edges$ check if the edge that $i$ is a reference to is already in $p$, or if the outvertices of that edge are already in $\text{value}(p)$ (meaning adding the edge would not extend the metapaths outputs). If either of these checks are true we do nothing. Otherwise, we add an entry in the dictionary $cur\_paths$ which adds a single edge to $p$, the dictionary entry has key $p \cup \{i\}$ with associated value $e_i \cup \text{value}(p) \text{ where }e_i \in E_2$. It is at this point that the consolidation of metapaths that contain the same set of edges happens. The dictionary $cur\_paths$ allows only one entry per key and paths that have the same edges have the same keys and values.  It is as fast to overwrite the data as it is to compare it.
    \item Set $prev\_paths$ to be  $cur\_paths$.
\end{enumerate}

The detail omitted here is how to expand the combined edges, we cover this in the next section.

\subsubsection{Expanding combined edges.}
Algorithm \ref{alg:ExpandMetapaths} describes the process of combining edges and operates as follows:
Because $E_2$ is a copy of $E_0$ that has the combined edges appended to it. Expanding the combined edges proceeds by iterating through the edge references in $path$ and appending any edge reference $i$  such that $i < |E_0|$ to the list that will be returned. If we encounter an edge in $path$ with a reference $i$ such that $i \ge |E_0|$ then we use the multi-dict $D$ to retrieve the list $edge\_ids$ which is the list of edge references to the edges that the edge referenced by $i$ replaced. With each $j$ in $edge\_ids$  we check if the outvertex of the edge referenced by $j$ is a subset of a set that is the union of every invertex in $path$. If so the edge $j$  is appended to the list that will be returned.

\begin{algorithm}
    \caption{Expand combined edges} \label{alg:ExpandMetapaths}
    \begin{algorithmic}[1]
        \Function{ExpandMetapaths}{$E_0$, $E_2$, $D$, $path$}

        \State $inputs \gets$ set $\bigcup_i \text{invertex}(e_i) \text{ where } e_i \in E_2 \text{ and }i \in path$
        \State $local\_edges \gets$ empty list
        \ForAll {$i$ in $path$}
        \If{$i < |E_0|$}
        \State append $i$ to $local\_edges$
        \Else
        \State $edge\_ids \gets$ list of $j$  mapped to $i$ in $D$
        \ForAll{$j \in edge\_ids$}
        \If{$\text{outvertex}(e_j) \cap inputs$ where $e_j \in E_2$}
        \State append $j$ to $local\_edges$
        \EndIf
        \EndFor
        \EndIf
        \EndFor
        \State\Return $local\_edges$
        \EndFunction
    \end{algorithmic}
\end{algorithm}

\subsection{Transitivity preserving Metagraph projection}

This projection uses dominant metapaths between a \textit{source} and \textit{target} in which no proper subset of the edges forms a metapath between the \textit{source} and \textit{target}.  Thus, the bulk of the work in generating a metagraph projection is in filtering the paths found by Algorithm~\ref{alg:getallpaths} in Section~\ref{sec:fam}. This algorithm finds all the edge dominant paths in which no proper subset of the edges forms a metapath between the \textit{source} and \textit{target} thus the dominant paths we seek are a subset of those found. Filtering these paths to isolate the dominant paths with the same subset property is simply a matter of selecting paths for which no other path found using Algorithm~\ref{alg:getallpaths} has an invertex set which is a proper subset of the path under consideration.

The projection generation algorithm, Algorithm~\ref{alg:projection}, operates by considering the set of generating elements, denoted as $generating\_subset$. For each element $x$ in this set, Algorithm~\ref{alg:getallpaths} is utilized to discover all paths connecting $generating\_subset \backslash \{x\}$ and $\{x\}$.
The paths are simplified into edges using Algorithm~\ref{alg:projection_edges}. The invertices of these edges are added to a set-trie called $tmp\_settrie$. This allows us to efficiently identify the input dominant edges, which are the ones in $tmp\_settrie$ without any subsets. The entries in $tmp\_settrie$ without subsets are then added to a set-trie multi-map $b$.
After identifying the dominant edges for each $i \in generator\_subset$, the set-trie multi-map $b$ contains all the edges representing the dominant paths. Each entry in $b$ associates an invertex (key) with an outvertex (value), forming the edges of the metagraph. The projection is completed by combining edges that share an invertex by forming an edge from the shared invertex to an outvertex which is the union of the outvertex sets. Finally, a metagraph object is returned that contains the $generating\_subset$ and any edges generated. The path simplification process in Algorithm~\ref{alg:projection_edges} is crucial as it determines the specific subset of elements from $generator\_subset$ that the path utilizes as a source. The details can be found in the following section.

\begin{algorithm}
    \caption{Projection}\label{alg:projection}
    \begin{algorithmic}[1]
        \Function{GetProjection}{$E_0, generator\_subset$}

        \ForAll{$i \in generator\_subset$}
        \State $inputs \gets$ \text{set of} $generator\_subset$
        \State $inputs = inputs \backslash \{i\}$
        \State $a \gets \text{empty list}$
        \ForAll{$x \in GetAllMetapathsFrom(E_0, inputs, \{i\})$}
        \State $\text{append } ProjectionEdges(E_0, x, inputs, generator\_subset) \text{ to }a$
        \EndFor

        \State $tmp\_settrie \gets$ \text{new SetTrie}
        \ForAll{$edge \in a$}
        \State $tmp\_settrie$.add($edge$.invertex)
        \EndFor

        \State $b \gets$ \text{new SetTrieMultiMap}

        \ForAll{$edge \in a$}
        \State $s \gets$ $tmp\_settrie$.subsets($edge$.invertex)
        \If{$|s| > 1$}
        \State \textbf{continue}
        \EndIf
        \State $b$.assign($edge$.invertex, $edge$.outvertex)
        \EndFor
        \EndFor

        \State $c \gets$ \text{empty list}
        \For{$edge$ \text{in} $b$.keys()}
        \State $outvertex \gets$ \text{empty set}
        \For{$o$ \text{in} $b$.get($edge$)}
        \State $outvertex \gets outvertex \cup o$
        \EndFor
        \State $\text{append Edge}(edge, outvertex) \text{ to }c$
        \EndFor

        \State $ret \gets$ \text{new Metagraph}($generator\_subset$)
        \State $ret$.add\_edges\_from($c$)
        \State \textbf{return} $ret$
        \EndFunction
    \end{algorithmic}
\end{algorithm}

\subsection{Edge simplification}
During the edge simplification process, the outvertex of the created edge is constrained to include only the pure outputs of the simplified metapath that are part of the $generating\_subset$. Similarly, the outvertex is limited to include only those elements from the $generating\_subset$ that are utilized by the metapath being simplified.

The outvertex is determined by taking the union of outvertex sets from all edges in the path, subtracting the union of invertex sets from all edges in the path, and then taking intersection of this result with the $generating\_subset$.

To calculate the outvertex, the edges need to be ordered in a sequence. This sequence ensures that the union of the source set and the outvertex sets of the preceding edges covers the invertex of the current edge. The process is similar to Algorithm~\ref{alg:GetSingleMetapathFrom} for finding a single path. The path is traversed in reverse, keeping track of the required outvertex sets at each step. Through this reverse traversal, the minimum subset of $generating\_subset$ necessary to complete the path is determined.

\begin{algorithm}
    \caption{Projection edges}\label{alg:projection_edges}
    \begin{algorithmic}[1]
        \Function{ProjectionEdges}{$E_0, path, source, generating\_subset$}
        \If{$|path| = 0$}
        \State \textbf{return}
        \EndIf
        \State $outputs \gets$ set $\bigcup_i \text{outvertex}(e_i) \text{ where } e_i \in E_0 \text{ and }i \in path$

        \State $outputs \gets outputs \backslash \bigcup_i \text{invertex}(e_i) \text{ where } e_i \in E_0 \text{ and }i \in path$
        \State $outvertex \gets outputs \cap generating\_subset$

        \State $bag \gets$ \text{set}(source)
        \State $new\_path \gets$ \text{empty list}
        \State $updated \gets$ \text{true}

        \While{$updated$}
        \State $updated \gets$ \text{false}

        \ForAll{$i$ \text{in} $path$}
        \If{$i$ \text{in} $new\_path$}
        \State \textbf{continue}
        \EndIf
        \If {$\text{invertex}(e_i) \subseteq bag \text{ where } e_i \in E_0$}
        \State append $i$ to $new\_path$
        \State $updated \gets$ \text{true}
        \State $bag \gets bag \cup \text{outvertex}(e_i) \text{ where } e_i \in E_0$
        \EndIf
        \EndFor
        \EndWhile

        \State $required \gets$ \text{set}(outvertex)
        \State $outvertices \gets$ \text{empty set}

        \ForAll{$i$ \text{in reversed} $new\_path$}
        \State $required \gets required \cup \text{invertex}(e_i) \text{ where } e_i \in E_0$
        \State $outvertices \gets outvertices \cup \text{outvertex}(e_i) \text{ where } e_i \in E_0$
        \State $required \gets required \backslash outvertices$
        \EndFor

        \State \textbf{return} \text{Edge}($required$, $outvertex$)
        \EndFunction
    \end{algorithmic}
\end{algorithm}

\subsection{Worked example of a TPP}
Using the metagraph shown in Figure~\ref{fig:tppmotivation}, we calculate the TPP of $\mathcal{H}=\langle\mathit{X,E} \rangle$ over $\mathit{X'=\{x_1, x_2, x_6, x_7, x_8\}}$, where $\mathit{X = \{x_1, x_2, x_3, x_4, x_5, x_6, x_7, x_8\}}$ and $\mathit{E = \{e_1, e_2, e_3, e_4, e_5\}}$. There are no dominant metapaths from any subset of $\mathit{X'}$ to either $\mathit{x_1}$ or $\mathit{x_2}$. Therefore, the set of dominant metapaths considered for the projection is as follows:

\begin{compactenum}
    \item $\mathit{M_1(\{x_1\}, \{x_6\}) = \{e_1,e_2 \}}$
    \item $\mathit{M_2(\{x_1, x_2\},\{x_7\}) = \{e_1,e_3,e_4 \}}$
    \item $\mathit{M_3(\{x_1, x_2\},\{x_8\}) = \{e_1,e_2,e_3,e_4,e_5 \}}$
    \item $\mathit{M_4(\{x_1, x_7\},\{x_8\}) = \{e_1,e_2,e_5 \}}$
    \item $\mathit{M_5(\{x_6, x_7\},\{x_8\}) = \{e_5\}}$
\end{compactenum}

The TPP selects dominant metapaths for edge creation based on a condition: none of the proper subsets of the edges from metapath $\mathit{M(V',{x'})}$ can form a dominant metapath $\mathit{M(U,{x'})}$ in $\mathcal{H}$, where $\mathit{U}\subseteq \mathit{X'}$. Consequently, metapath $\mathit{M_3({x_1, x_2},{x_8}) = \{e_1,e_2,e_3,e_4,e_5 \}}$ is excluded because it includes $e_5$, which alone forms $\mathit{M_5({x_6, x_7},{x_8}) = \{e_5\}}$. Similarly, $\mathit{M_4({x_1, x_7},{x_8}) = \{e_1,e_2,e_5 \}}$ is also excluded. This leaves the remaining metapaths to be combined as follows:

\begin{compactenum}
    \item $\mathit{C(e'_1) =\{ M_1(\{x_1\}, \{x_6\})\}}$
    \item $\mathit{C(e'_2) =\{M_2(\{x_1, x_2\},\{x_7\})\}}$
    \item $\mathit{C(e'_3) =\{M_5(\{x_6, x_7\},\{x_8\})\}}$
\end{compactenum}

Leading to the following edges:
\begin{compactenum}
    \item $\mathit{e'_1 = (\{x_1\}, \{x_6\})}$
    \item $\mathit{e'_2 = (\{x_1, x_2\},\{x_7\})}$
    \item $\mathit{e'_3 = (\{x_6, x_7\},\{x_8\})}$
\end{compactenum}

\begin{figure}[!htb]
    \centering
    \begin{tikzpicture}
        \node[minimum width=1cm,draw,circle] (1) at (0,3.5) {$\mathit{x_1}$};
        \node (2) at (0,1){$\mathit{x_2}$};
        \node[ellipse, minimum width = 1.8cm, minimum height = 4
            cm, draw] (12) at (0,2.25){};

        \node[minimum width=1cm,draw,circle] (6) at (5,3.5) {$\mathit{x_6}$};
        \node[minimum width=1cm,draw,circle] (7) at (5,1.5){$\mathit{x_7}$};

        \node[ellipse, minimum width = 2cm, minimum height = 4cm, draw] (11) at (5,2.5){};
        \node[minimum width=1cm,draw,circle] (8) at (7.5,2.5){$\mathit{x_8}$};

        \draw[-{Latex[length=2mm]}] (1) -- (6) node [midway, above, sloped] (TextNode) {$e'_1$};
        \draw[-{Latex[length=2mm]}] (12) -- (7) node [midway, above, sloped] (TextNode) {$e'_2$};
        \draw[-{Latex[length=2mm]}] (11) -- (8) node [midway, above, sloped] (TextNode) {$e'_3$};
    \end{tikzpicture}
    \caption{Transitivity preserving projection of Figure~\ref{fig:tppmotivation}.}
    \label{fig:tpp}
\end{figure}
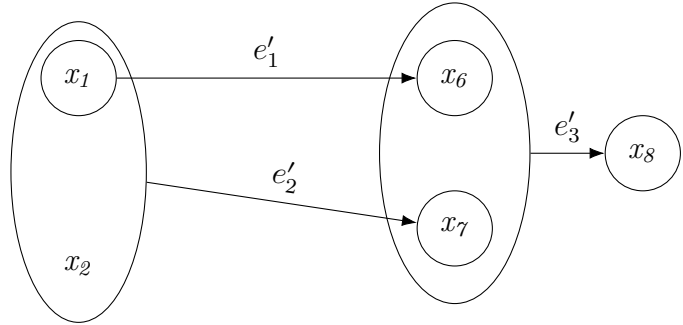

Illustrated in Figure~\ref{fig:tpp} is the TPP of the metagraph depicted in Figure~\ref{fig:tppmotivation}, $\mathit{\mathcal{H}=\langle\mathit{X,E}\rangle}$, where $\mathit{X = \{x_1, x_2, x_3, x_4, x_5, x_6, x_7, x_8\}}$ and $\mathit{E = \{e_1, e_2, e_3, e_4, e_5\}}$, projected over $\mathit{X'=\{x_1, x_2, x_6, x_7, x_8\}}$. Compared to the metagraph projection in Figure~\ref{fig:projection}, the TPP is visually simpler.

A more striking comparison is observed between the TPP $\mathit{\mathcal{H'}_2 = \langle X'_2, E'_2 \rangle}$ (Figure~\ref{fig:H_2tpp}) of metagraph $\mathit{\mathcal{H}_2 = \langle X_2, E_2 \rangle}$ (Figure~\ref{fig:H_2}) with $\mathit{X_2 = \{F_0, A_0, A_1, A_2, B_0, B_1, B_2, C_1, C_2, D_1, D_2, E_1, E_2 \}}$ and $\mathit{E_2 = \{e_{1,2}, e_{2,2}, e_{3,2}, e_{4,2}, e_{1,1}, e_{2,1}, e_{3,1}, e_{4,1}, e_0 \} }$, projected over $\mathit{X'_2 = \{F_0, A_0, A_1, A_2, B_0, B_1, B_2 \}}$. The TPP is a significant visual improvement over the metagraph projection for $\mathit{\mathcal{H}_2}$  shown in Figure~\ref{fig:H_2p}.

\begin{figure}[!htb]
    \begin{center}
    \resizebox{0.9\textwidth}{!}{
    \begin{tikzpicture}
        \node[minimum width=1cm,draw,circle] (1) at (0,3.5) {$\mathit{A_2}$};
        \node (2) at (0,1){$\mathit{B_2}$};
        \node[minimum width=1cm,draw,circle] (6) at (5,3.5) {$\mathit{A_1}$};
        \node[minimum width=1cm,draw,circle] (7) at (5,1.5){$\mathit{B_1}$};
        \node[minimum width=1cm,draw,circle] (11) at (10,3.5) {$\mathit{A_0}$};
        \node[minimum width=1cm,draw,circle] (12) at (10,1.5){$\mathit{B_0}$};
        \node[minimum width=1cm,draw,circle] (13) at (12.5,2.5){$\mathit{F_0}$};
        \node[ellipse, minimum width = 2cm, minimum height = 4cm, draw] (1-2) at (0,2.5){};
        \node[ellipse, minimum width = 2cm, minimum height = 4cm, draw] (67) at (5,2.5){};
        \node[ellipse, minimum width = 2cm, minimum height = 4cm, draw] (1112) at (10,2.5){};

        \draw[-{Latex[length=2mm]}] (1) -- (6) node [midway, above, sloped] (TextNode) {$e'_1$};
        \draw[-{Latex[length=2mm]}] (1-2) -- (7) node [midway, above, sloped] (TextNode) {$e'_2$};
        \draw[-{Latex[length=2mm]}] (6) -- (11) node [midway, above, sloped] (TextNode) {$e'_3$};
        \draw[-{Latex[length=2mm]}] (67) -- (12) node [midway, above, sloped] (TextNode) {$e'_4$};
        \draw[-{Latex[length=2mm]}] (1112) -- (13) node [midway, above, sloped] (TextNode) {$e'_5$};
    \end{tikzpicture}
    }
    \end{center}
    \caption{$\mathit{\mathcal{H}'_2}$ from the transitivity preserving projection of Figure~\ref{fig:H_2}}

    \label{fig:H_2tpp}
\end{figure}
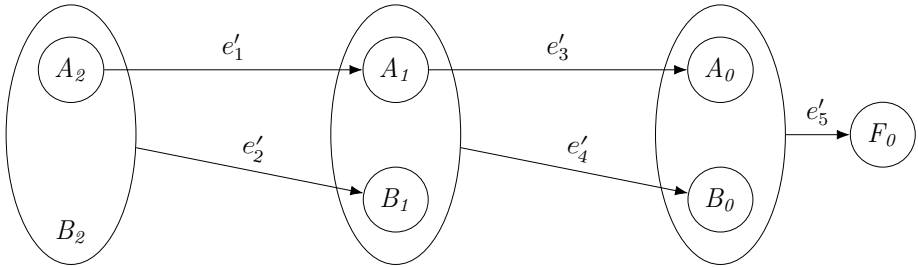

We deduce that the TPP $\mathcal{H}'_n$ of $\mathcal{H}_n$ (as defined in Lemma \ref{complex}) has $\mathit{2n + 1}$ edges, which is always less than the $\mathit{4n + 1}$ edges in $\mathcal{H}_n$.

\subsection{TPP complexity}

In our discussion on the complexity of the metagraph projection, we recognized the need to capture relationships between elements of $\mathit{X'}$ resulting from edges involving elements in $\mathit{X\backslash X'}$ using a form of transitive closure. The TPP employs the same closure mechanism as the metapath projection, but solely for capturing relationships involving edges that involve elements from $\mathit{X\backslash X'}$. Because it only calculates these relationships using dominant paths that part of the projection is source-minimum in the sense defined in \citet{ausiello1986} (i.e.\ there is no B-hypergraph equivalent to $\mathcal{H}$ with fewer source sets).

The TPP reduces relationships between elements of $\mathit{X'}$ resulting from edges not involving elements in $\mathit{X\backslash X'}$ to the dominant edges that existed between these elements in the metagraph. As these edges would also be present in the metagraph projection, it can be deduced that the number of edges in the TPP is always less than or equal to the number of edges in the metagraph projection.

A transitive closure can potentially have more edges than the original metagraph, meaning that the TPP may contain more edges than its preimage depending on the exact geometry of the underlying metagraph. But this is also the case with the metagraph projection.

In terms of edge count the worst case performance of the TPP is precisely the same as that of the metagraph projection but in all cases it has an advantage in computation as we demonstrate in the following section.

\subsection{Algorithmic complexity}
The algorithm for the construction of the metagraph projection is as follows:
\begin{enumerate}
    \item For every $\mathit{x'\in X'}$, for every $\mathit{V'\subseteq X'\backslash\{x\}}$ find all the dominant paths $\mathit{M(V',\{x'\})}$ in $\mathcal{H}$ then add them to a set of dominant metapaths $\mathit{\{M_1(V',\{x'\}),\ldots\}}$;
    \item For every unique invertex $\mathit{V'}$ of the metapaths in $\mathit{\{M_1(V',\{x'\}),\ldots\}}$ construct an edge $\mathit{e' = \langle V', W' \rangle \in E'}$ where $\mathit{W' = \{\bigcup_{i}\,x_i |\exists\; M_1(V',\{x_i\}) \in \{M_1(V',\{x'\}),\ldots\}\}}$.
\end{enumerate}

The significant part here is that for each $\mathit{x'\in X'}$ there are $\mathit{2^{|X'|-1}}$ possible $\mathit{V'}$. Finding dominant paths is an NP hard problem \citep{ward2023} that currently involves a brute force solution. Doing this $\mathit{|X'|2^{|X'|-1}}$ times is prohibitive for anything but small $\mathit{|X'|}$.

Creating the TPP is seemingly as prohibitive as its definition filters the same set of dominant metapaths as the metagraph projection. However, in this case, the filtering can be performed before identifying the required dominant metapaths. The algorithm for constructing the TPP is as follows:

\begin{enumerate}
    \item For every $\mathit{x'\in X'}$, find all the edge dominant paths from $\mathit{M(X' \backslash \{x'\},\{x'\})}$ in $\mathcal{H}$ then filter them to find only the dominant metapaths and add these to a set of dominant metapaths $\mathit{\{M_1(V',\{x'\}),\ldots\}}$;
    \item For every unique invertex $\mathit{V'}$ of the metapaths in $\mathit{\{M_1(V',\{x'\}),\ldots\}}$ construct an edge $\mathit{e' = \langle V', W' \rangle \in E'}$ where $\mathit{W' = \{\bigcup_{i}\,x_i |\exists\; M_1(V',\{x_i\}) \in \{M_1(V',\{x'\}),\ldots\}\}}$.
\end{enumerate}

Thus, the exhaustive search for metapaths occurs only $\mathit{|X'|}$ times. Filtering edge dominant metapaths to find dominant metapaths can be accomplished in polynomial time using a settrie data structure. This structure stores the invertices of all the edge dominant paths and allows checking each path against the settrie to determine if it contains an invertex that is a subset of the current path's invertex. If such a subset is found, the current path is rejected.

Therefore, while the TPP relies on the solution of an NP-hard problem, it provides a clear advantage in terms of computation time compared to the metagraph projection.

\section{Implementation and Results}

\subsection{Data Structure}

We use a relatively new data structure  called a set-trie multimap \citep{Savnik2013} to direct this search. Set-tries are efficient structures for finding subsets and supersets of a given set. Our algorithms use set-tries in the following way:

\begin{enumerate}
    \item Given a list of metagraph edges $E$ then for every $e_i \in E$ the key-value $invertex(e_i)$ is mapped to $i$ in the set-trie $S$.  Thus given a set $B$ of elements it is possible to efficiently find by reference all edges that have invertices wholly within $B$ by simply calling the subset function of $S$ passing it the set $B$.
    \item Given a list of metagraph edges $E$ then for every $e_i \in E$ the key-value $invertex(e_i)$ is stored in the set-trie $S$.  Thus given an edge $e$ it is possible to efficiently determine if there is an edge $e_j \in E$ that has a smaller invertex set. This is used in the context of metagraph projections. Where all the edges in $E$ have the same outvertex set, and we are looking for a dominant edge.
\end{enumerate}


\subsection{Experiments}

We present here the results of our new algorithms on a few real-world examples. We also cross compare them with previous works. We apply our TPP to two real-world examples of metagraphs: (1) Aircraft supply chain metagraph~\cite{hopp2023}, and (2) a security metagraph~\cite{ranathunga2020verifiable}.

\subsubsection{Aircraft supply chain Metagraph}

In~\cite{hopp2023}, Hopp et al. explore the application of a conditional metagraph to enhance cybersecurity maturity in the aviation supply chain, focusing on the landing gear supply chain. The metagraph represents the supply chain as a network of nodes (suppliers and components) and edges (relationships and information exchanges), with conditional attributes such as cybersecurity maturity levels, enabling algebraic calculations to identify secure paths and analyze breaches.  The graph in \autoref{fig:aircraftmetagraph} is the metagraph of the example that was provided in~\cite{hopp2023}.

\begin{figure}[htbp]
    \centering
    \includegraphics[width=\linewidth]{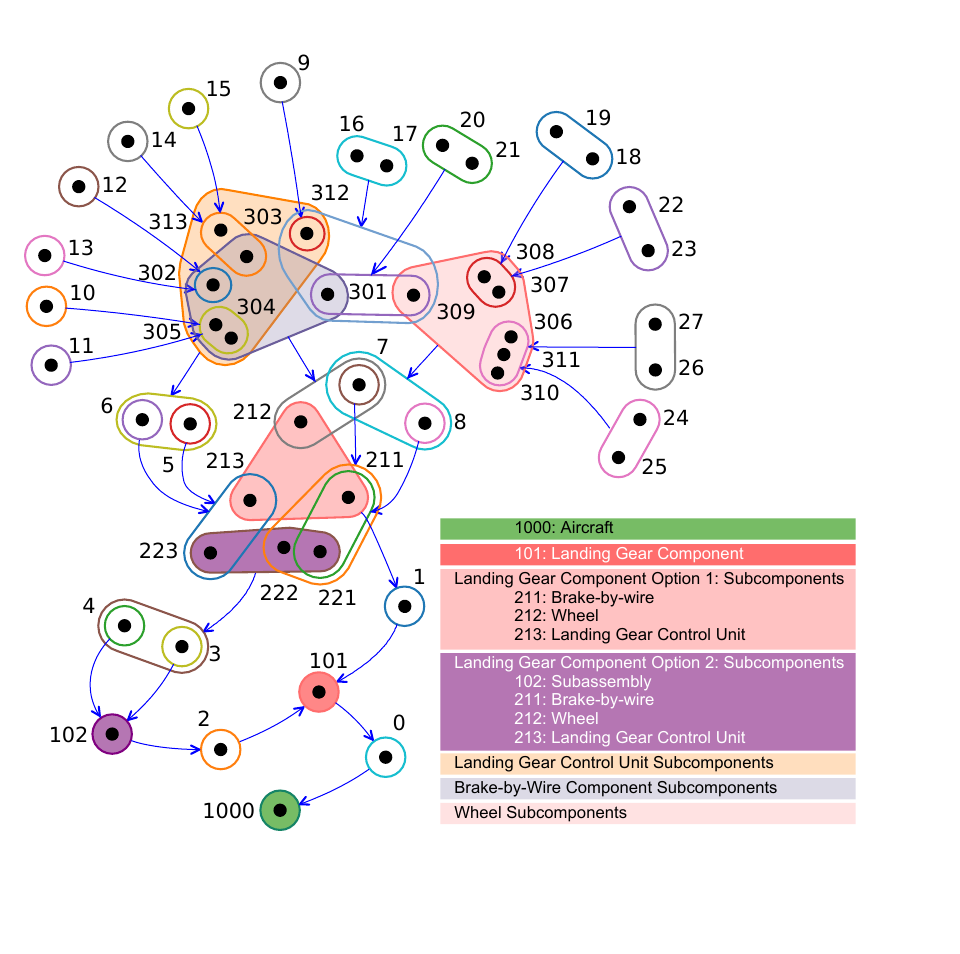}

    \caption{Metagraph representation of a supply chain for an aircraft part, adapted from \citet{hopp2023}. For visual clarity, helper node 9999 is not shown. In the original, node 9999 serves as the invertex for 14 metaedges with outvertex sets \{9\}, \{10\}, \{11\}, \{12\}, \{13\}, \{14\}, \{15\}, \{16, 17\}, \{18, 19\}, \{20, 21\}, \{22, 23\}, \{24, 25\}, and \{26, 27\}.\label{fig:aircraftmetagraph}}

\end{figure}

\begin{figure}[htbp]
    \centering
    \includegraphics[width=\linewidth]{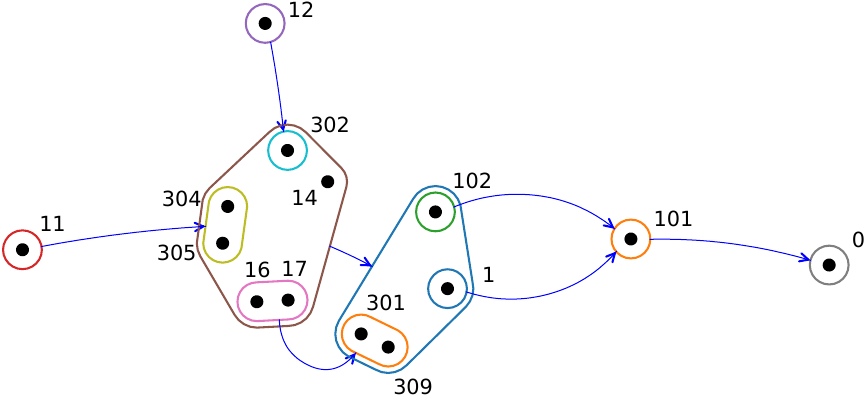}
    \hspace*{-2cm}  
    \caption{Projection of Metagraph of a supply chain for an aircraft part projected over $\mathit{X' = \{0, 1, 11, 12, 14, 16, 17, 101, 102.  301, 302, 304, 305, 309\}}$\label{fig:aircraftmetagraphprojection}.}
\end{figure}


It was shown in~\cite{hopp2023} that the metagraph model enhances cybersecurity by filtering suppliers based on maturity, calculating secure sourcing alternatives, supporting rapid incident response through attack propagation tracing, and integrating global standards like CMMC, NIST, and ISO. However, the metagraph faces scalability limitations due to its inability to efficiently handle large-scale, industry-wide supply chain data, as the computational complexity of metapath calculations grows with the number of nodes and edges. Additionally, visualization challenges arise from the need to balance transparency with competitive advantage concerns, requiring projections to conceal sensitive supplier information, which complicates practical implementation.



The algorithm for finding all paths of \citet{basu2007metagraphs} and \citet{blanning1997} exhausted 16GB of memory over the course of an hour attempting to find all the paths between nodes $9999$ and $1000$. Their algorithm for calculating projections did not terminate in 24 hours when attempting to calculate the projection over the element set $\mathit{X' = \{ 9,..,27, 303, 211, 212, 213, 223, 222, 221\}}$. Our algorithm for finding all paths finds all 5248 edge dominant paths (in this case they are all also dominant paths) in less than 2 minutes. Likewise, a TPP projection over $\mathit{X' = \{ 9,..,27, 303, 211, 212, 213, 223, 222, 221\}}$ takes less than 0.2 seconds to complete.   \autoref{fig:aircraftmetagraphprojection} shows a projection over $\mathit{X' = \{0, 1, 11, 12, 14, 16, 17, 101, 102.  301, 302, 304, 305, 309\}}$ which also took less than 0.2 seconds to complete.

\subsubsection{Security Metagraph}
\autoref{fig:UoA_Metagraph_full} illustrates the University of Adelaide network represented as a metagraph,  first presented in~\cite{ranathunga2020verifiable}. The case study in~\cite{ranathunga2020verifiable} evaluates the application of metagraph-based policy modeling and verification in a large-scale university network. The network comprises over 30,000 users, 43 logical zones, and hundreds of devices, including firewalls and routers. Using a custom-built parser and the MGtoolkit, the authors extracted and modeled over 1,000 high-level policy rules across multiple domains such as access control, QoS, intrusion detection, and URL filtering.  A key highlight of the study is the use of metagraph projections, exemplified in \cite{ranathunga2020verifiable}. A more detailed version of the network is shown in \autoref{fig:UoA_Metagraph_full}. In~\cite{ranathunga2020verifiable}, the authors showed that by using projection, they significantly simplified complex policy graphs by focusing on relevant subsets of the network. This visualization capability proved especially valuable for administrators in understanding and managing policy reachability and conflict resolution, outperforming traditional tools in clarity and effectiveness. The projection in~\cite{ranathunga2020verifiable} was done manually and on a specific subset of nodes.

\autoref{fig:UoA_Metagraph_No_DMZ} demonstrates a TPP projection of \autoref{fig:UoA_Metagraph_full} over the generating set, excluding blacklist\_IPs, subzone\_exception and exception. This calculation takes approximately 0.1 second to execute. In comparison, the algorithm proposed by \citet{basu2007metagraphs} consumed 16GB of memory in 15 minutes when attempting their projection method on the same set and did not terminate within 24 hours. Additionally,  TTP projection of \autoref{fig:UoA_Metagraph_No_DMZ} over the generating set, excluding the two exception sets, takes around 0.1 seconds to compute, while the algorithm by Blanning et al \cite{basu2007metagraphs} exhausts memory and again yields no solution within 24 hours.


These preliminary results confirm the validity of our approach. We have discovered that the definition of a metapath can be detached from the concept of a simple-path.
Additionally, we have devised a polynomial time algorithm for identifying individual paths, which proves useful for determining connectivity and dominance. Further, our find all paths and projection algorithm now possess the scalability required to address some real-world problems modeled as metagraphs.

\begin{sidewaysfigure}[htbp]
    \centering
    \includegraphics[width=12cm]{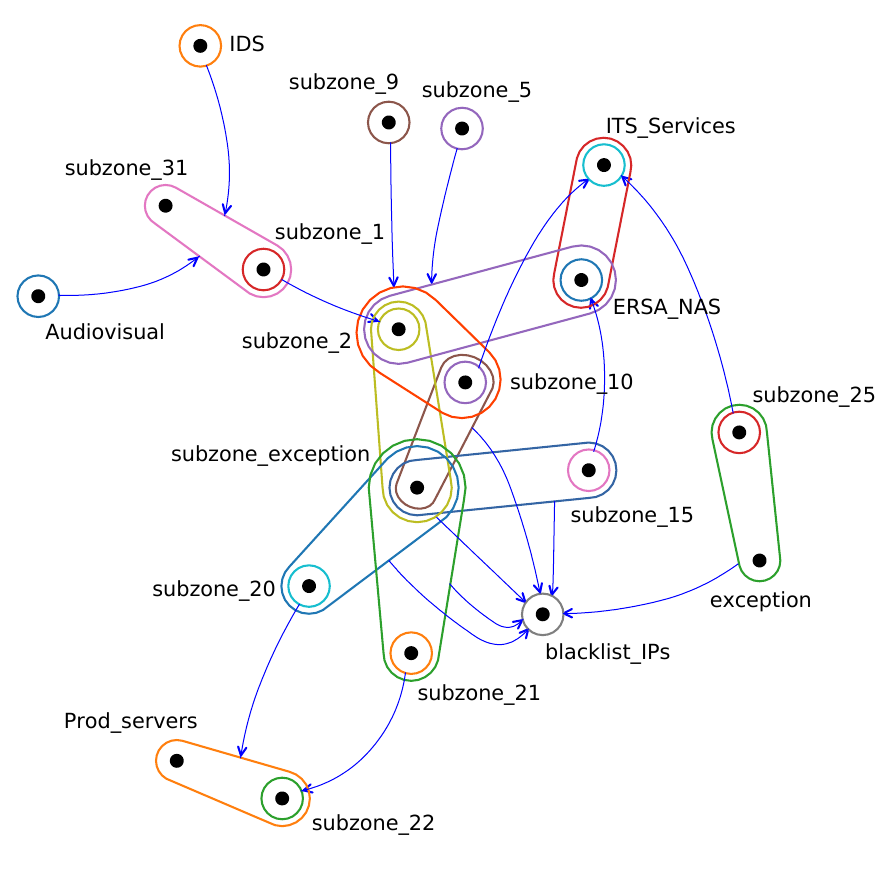}
    \caption{University of Adelaide metagraph.}\label{fig:UoA_Metagraph_full}
\end{sidewaysfigure}
\begin{sidewaysfigure}[htbp]
    \centering
    \includegraphics[width=12cm]{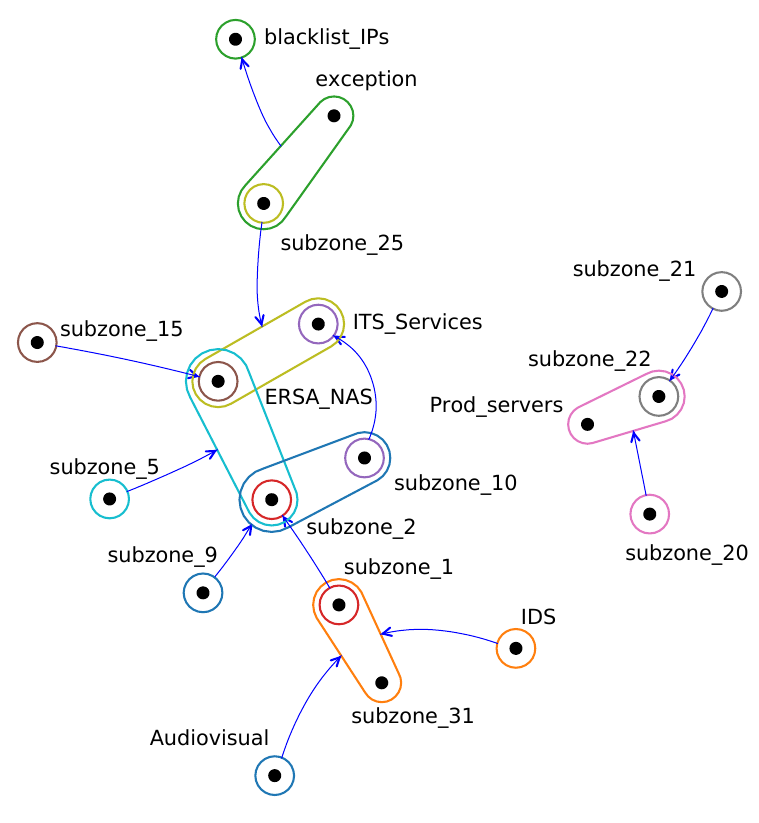}
    \caption{University of Adelaide metagraph projection without Exceptions for Blacklisted IPs zones.}\label{fig:UoA_Metagraph_No_DMZ}
\end{sidewaysfigure}


It should be noted that for finding all paths between $source$ and $target$ sets, the search space is exponential in the number of edges in the metagraph. \autoref{alg:getallpaths} exploits every opportunity to minimize the number of edges being considered in the search space but at present does not employ any form of multiprocessing. There is an opportunity to parallelize this algorithm. However, if we increase the number of edges to be considered by one, the search space will double. To effectively handle this increased workload, we require a significant level of parallelization, such as that which can be achieved through the use of GPUs.
The breadth-first search algorithm for finding all paths is suitable for a CUDA implementation. Existing CUDA implementations of breadth-first search exist for classic graphs structures, such as the one by \citet{harish2007}, and a similar scheme could be implemented for metagraphs.
\section{Conclusions}

This study establishes a formal connection between the metagraph projection defined by \citet{basu2007metagraphs} and the succinct transitive closure of hypergraphs introduced by \citet{ausiello2016}. Contrary to its intended purpose, we demonstrate that the metagraph projection can, in fact, increase structural complexity. Through a detailed examination of its construction, we propose the Transitivity Preserving Projection (TPP), a novel alternative that retains essential transitive relationships while offering a provably unique and minimal representation. The TPP not only enhances visual clarity but also significantly improves algorithmic efficiency, thereby advancing the state of the art in metagraph simplification and visualization.

The algorithms introduced in this work extend the practical boundary of metagraph analysis, enabling the handling of larger and more complex instances than previously feasible. However, further empirical and theoretical work is needed to rigorously characterize their computational complexity, particularly given the intricate dependencies between vertex sets and edge structures. These algorithms also lay the groundwork for re-implementing other core metagraph operations. Notably, the current exponential-time algorithm for identifying dominant paths can now be replaced with a polynomial-time alternative, thereby establishing that the NP-hard problem of finding dominant metapaths \citep{ward2023} is also NP-complete.

\section*{Acknowledgements}

The Commonwealth of Australia (represented by the Advanced Strategic Capabilities Accelerator Emerging and Disruptive Technologies program of the Defence Science and Technology Group) supports this research through a Defence Science Partnerships agreement. The first author also gratefully acknowledges the support of the University of Adelaide in the form of a University of Adelaide scholarship. 

\bibliography{mgtoolkit-algorithms}

@ARTICLE{10290912,
  author={Oliver, Peter and Zhang, Eugene and Zhang, Yue},
  journal={IEEE Transactions on Visualization and Computer Graphics}, 
  title={Scalable Hypergraph Visualization}, 
  year={2024},
  volume={30},
  number={1},
  pages={595-605},
  doi={10.1109/TVCG.2023.3326599}}

@ARTICLE{10673783,
  author={Oliver, Peter and Zhang, Eugene and Zhang, Yue},
  journal={IEEE Transactions on Visualization and Computer Graphics}, 
  title={Structure-Aware Simplification for Hypergraph Visualization}, 
  year={2024},
  volume={},
  number={},
  pages={1-10},
  doi={10.1109/TVCG.2024.3456367}}

@ARTICLE{9721603,
  author={Zhou, Youjia and Rathore, Archit and Purvine, Emilie and Wang, Bei},
  journal={IEEE Transactions on Visualization and Computer Graphics}, 
  title={Topological Simplifications of Hypergraphs}, 
  year={2023},
  volume={29},
  number={7},
  pages={3209-3225},
  doi={10.1109/TVCG.2022.3153895}}

@Article{math10050759,
AUTHOR = {Molnár, Bálint and Benczúr, András},
TITLE = {The Application of Directed Hyper-Graphs for Analysis of Models of Information Systems},
JOURNAL = {Mathematics},
VOLUME = {10},
YEAR = {2022},
NUMBER = {5},
ARTICLE-NUMBER = {759},
URL = {https://www.mdpi.com/2227-7390/10/5/759},
ISSN = {2227-7390},
DOI = {10.3390/math10050759}
}

@article{doi:10.1089/cmb.2023.0242,
author = {Krieger, Spencer and Kececioglu, John},
title = {Shortest Hyperpaths in Directed Hypergraphs for Reaction Pathway Inference},
journal = {Journal of Computational Biology},
volume = {30},
number = {11},
pages = {1198-1225},
year = {2023},
doi = {10.1089/cmb.2023.0242},
    note ={PMID: 37906100},
URL = {     
        https://doi.org/10.1089/cmb.2023.0242
},
eprint = { 
    
        https://doi.org/10.1089/cmb.2023.0242
    
}
}

@InProceedings{10.1007/3-540-45446-2_20,
author="Ausiello, Giorgio
and Franciosa, Paolo G.
and Frigioni, Daniele",
title="Directed Hypergraphs: Problems, Algorithmic Results, and a Novel Decremental Approach",
booktitle="Theoretical Computer Science",
year="2001",
publisher="Springer Berlin Heidelberg",
address="Berlin, Heidelberg",
pages="312--328",
isbn="978-3-540-45446-5"
}

@InProceedings{parsonage2024,
author="Eric Parsonage
and Max Ward
and Hung Nguyen",
title="Higher-Order Graph Models in Network Security ",
booktitle="Proceedings of  the 14th edition of International Workshop on Resilient Networks Design and Modeling",
year="2024",
}

@article{gallo1993directed,
	author = {Gallo, Giorgio and Longo, Giustino and Pallottino, Stefano and Nguyen, Sang},
	journal = {Discrete Applied Mathematics},
	number = {2-3},
	pages = {177--201},
	publisher = {Elsevier},
	title = {Directed hypergraphs and applications},
	volume = {42},
	year = {1993}}

@book{basu2007metagraphs,
	author = {Basu, Amit and Blanning, Robert W},
	publisher = {Springer Science \& Business Media},
	title = {Metagraphs and Their Applications},
	volume = {15},
	year = {2007}}

@article{ranathunga2020verifiable,
	author = {Ranathunga, Dinesha and Roughan, Matthew and Nguyen, Hung},
	journal = {IEEE Transactions on Dependable and Secure Computing},
	publisher = {IEEE},
	title = {Verifiable policy-defined networking using metagraphs},
	year = {2020}}

@article{ward2023,
	author = {Max Ward and Lo{\"\i}c Miller and Reynaldo Gil-Pons and Matthew Roughan and Hung X. Nguyen},
	howpublished = {unpublished},
	title = {Finding Input-Dominant Hyperpaths and Metapaths is NP-Hard},
	year = {2023}}

@article{ausiello2016,
	author = {Ausiello, Giorgio and Laura, Luigi},
	doi = {10.1016/j.tcs.2016.03.016},
	journal = {Theoretical Computer Science},
	month = {03},
	title = {Directed Hypergraphs: Introduction and fundamental algorithms---a survey},
	volume = {658},
	year = {2016}}

@article{blanning1997,
	author = {Amit Basu and Robert W. Blanning and Avraham Shtub},
	issn = {00251909, 15265501},
	journal = {Management Science},
	number = {5},
	pages = {623--639},
	publisher = {INFORMS},
	title = {Metagraphs in Hierarchical Modeling},
	url = {http://www.jstor.org/stable/2634400},
	urldate = {2023-05-16},
	volume = {43},
	year = {1997}}

@article{ausiello1986,
	author = {Ausiello, G. and D'Atri, A. and Sacc\`{a}, D.},
	doi = {10.1137/0215029},
	eprint = {https://doi.org/10.1137/0215029},
	journal = {SIAM Journal on Computing},
	number = {2},
	pages = {418-431},
	title = {Minimal Representation of Directed Hypergraphs},
	url = {https://doi.org/10.1137/0215029},
	volume = {15},
	year = {1986}}

@mastersthesis{hopp2023,
	author = {Luisa Caretta Hopp},
	month = {May},
	school = {Tallinn University of Technology},
	title = {APPLICATION OF THE METAGRAPH TO THE AVIATION SUPPLY CHAIN: A CASE STUDY},
	type = {Masters Thesis},
	year = {2023}}

@inproceedings{harish2007,
	author = {Pawan Harish and P J Narayanan},
	booktitle = {International Conference on High Performance Computing},
	title = {Accelerating Large Graph Algorithms on the GPU Using CUDA},
	year = {2007}}

@inproceedings{Savnik2013,
	author = {Iztok Savnik},
	booktitle = {CD-ARES},
	title = {Index Data Structure for Fast Subset and Superset Queries},
	year = {2013}}
\bibliographystyle{agsm}
\end{document}